\documentclass[journal]{IEEEtran}

\ifCLASSINFOpdf

\else

\fi

\ifCLASSOPTIONcompsoc
  \usepackage[caption=false,font=normalsize,labelfont=sf,textfont=sf]{subfig}
\else
  \usepackage[caption=false,font=footnotesize]{subfig}
\fi

\usepackage{amsthm}

\usepackage{amssymb}
\usepackage{mathtools}
\usepackage{cases}
\usepackage{autobreak}

\usepackage{booktabs}
\usepackage{multirow}

\usepackage{color}
\usepackage{hyperref}
\usepackage{cite}

\usepackage{amsmath}
\usepackage{extarrows}

\usepackage{lipsum}
\usepackage{caption}

\usepackage{comment}

\usepackage{algorithmic}
\usepackage{algorithm}
\usepackage{bbm}
\usepackage{bm}

\newcommand{\E}{\mathbb{E}}

\newcommand{\cov}{\mathrm{Cov}}

\newcommand{\diag}{\mathrm{diag}}

\newcommand{\LR}{\bold{R}^{\frac{1}{2}}}

\newcommand{\RT}{\bold{T}^{\frac{1}{2}}}

\newcommand{\FS}{{\bold{S}}}
\newcommand{\LC}{\bold{C}^{\frac{+}{2}}}
\newcommand{\RC}{\bold{C}^{\frac{-}{2}}}
\newcommand{\FC}{{\bold{C}}}
\newcommand{\LF}{\bold{F}^{\frac{1}{2}}}
\newcommand{\MC}{\bold{C}^{\frac{1}{2}}}

\newcommand{\FF}{\bold{F}}
\newcommand{\Bs}{\bold{s}}

\newcommand{\FR}{{\bold{R}}}
\newcommand{\FT}{{\bold{T}}}
\newcommand{\BK}{{\bold{K}}}
\newcommand{\BP}{{\bold{P}}}
\newcommand{\BU}{{\bold{U}}}
\newcommand{\RU}{{\bold{U}^{\frac{1}{2}}}}

\newcommand{\BQ}{{\bold{Q}}}
\newcommand{\BH}{{\bold{H}}}

\newcommand{\BI}{{\bold{I}}}
\newcommand{\BX}{{\bold{X}}}

\newcommand{\BZ}{{\bold{Z}}}
\newcommand{\BY}{{\bold{Y}}}

\newcommand{\BPS}{{\bold{\Psi}}}
\newcommand{\BG}{{\bold{G}}}

\newcommand{\BT}{{\bold{T}}}
\newcommand{\Bh}{{\bold{h}}}

\newcommand{\Bxi}{{\bm{\xi}}}

\newcommand{\By}{{\bold{y}}}
\newcommand{\Bw}{{\bold{w}}}

\newcommand{\BA}{{\bold{A}}}
\newcommand{\BB}{{\bold{B}}}
\newcommand{\BC}{{\bold{C}}}

\newcommand{\BW}{{\bold{W}}}
\newcommand{\BM}{{\bold{M}}}
\newcommand{\BL}{{\bold{L}}}

\newcommand{\BF}{{\bold{F}}}

\newcommand{\BO}{{\mathcal{O}}}

\newcommand{\BXi}{{\bold{\Xi}}}

\DeclareMathOperator{\Tr}{Tr}

\newcommand{\RNum}[1]{\uppercase\expandafter{\romannumeral #1\relax}}

\newtheorem{remark}{Remark}
\newtheorem{theorem}{Theorem}
\newtheorem{lemma}{Lemma}
\newtheorem{proposition}{Proposition}

\newcommand\numberthis{\addtocounter{equation}{1}\tag{\theequation}}

\allowdisplaybreaks[4]
\begin{document}

\title{Fluid Antenna Meets RIS: Random Matrix Analysis and Two-Timescale Design for Multi-User Communications}


\author{Xin~Zhang, Dongfang Xu, Jingjing~Wang,~\IEEEmembership{Senior Member,~IEEE},~Shenghui~Song,~\IEEEmembership{Senior Member,~IEEE}, 
       \\ 
       Derrick Wing Kwan Ng,~\IEEEmembership{Fellow, IEEE}, and M\'erouane Debbah,~\IEEEmembership{Fellow, IEEE}
\thanks{

X. Zhang and J. Wang are with the School of Cyber Science and Technology, Beihang University, Beijing 100191, China. (\{zhangxinn, drwangjj\}@buaa.edu.cn)

D. Xu is with the Department of Electronic and Computer Engineering, The Hong Kong University of Science and Technology, Hong Kong.
(eedxu@ust.hk)

S. H. Song is  with the Division of Integrative Systems and Design and the Department of Electronic and Computer Engineering at the Hong Kong University of Science and Technology, Hong Kong. (eeshsong@ust.hk)

D. W. K. Ng is with the School of Electrical Engineering and Telecommunications, University of New South Wales, Sydney, NSW 2052,
Australia (e-mail: w.k.ng@unsw.edu.au)

M. Debbah is with  KU 6G Research Center, Department of Computer and Information Engineering, Khalifa University, Abu Dhabi 127788, UAE (email: merouane.debbah@ku.ac.ae) and also with CentraleSupelec, University Paris-Saclay, 91192 Gif-sur-Yvette, France.
}}

\IEEEaftertitletext{\vspace{-2\baselineskip}}
\maketitle

\begin{abstract} The reconfigurability of fluid antenna systems (FASs) and reconfigurable intelligent surfaces (RISs) provides significant flexibility in optimizing channel conditions by jointly adjusting the positions of fluid antennas and the phase shifts of RISs. However, it is challenging to acquire the instantaneous channel state information (CSI) for both fluid antennas and RISs, while frequent adjustment of antenna positions and phase shifts will significantly increase the system complexity. To tackle this issue, this paper investigates the two-timescale design for FAS-RIS multi-user systems with linear precoding, where only the linear precoder design requires instantaneous CSI of the end-to-end channel, while the FAS and RIS optimization relies on statistical CSI. The main challenge comes from the complex structure of channel and inverse operations in linear precoding, such as regularized zero-forcing (RZF) and zero-forcing (ZF). Leveraging on random matrix theory (RMT), we first investigate the fundamental limits of FAS-RIS systems with RZF/ZF precoding by deriving the ergodic sum rate (ESR). This result is utilized to determine the minimum number of selected antennas to achieve a given ESR. Based on the evaluation result, we propose an algorithm to jointly optimize the antenna selection, regularization factor of RZF, and phase shifts at the RIS. Numerical results validate the accuracy of performance evaluation and demonstrate that the performance gain brought by joint FAS and RIS design is more pronounced with a larger number of users.
\end{abstract}


\IEEEpeerreviewmaketitle

\section{Introduction}

Over the past two decades, multiple-input multiple-output (MIMO) technology has emerged as a key enabler for enhancing both the throughput and reliability of wireless communication systems~\cite{lu2014overview,liu10845870}. It is widely recognized that increasing the number of antennas can lead to higher spectral efficiency due to greater spatial multiplexing gains. However, the rising hardware costs and increased radio-frequency (RF) energy consumption present significant challenges in deploying a large number of RF chains. Consequently, selecting antenna positions that optimize channel conditions becomes essential. To address this, Wong~\textit{et al.} introduced fluid antenna systems (FASs) featuring position-flexible antennas that leverage intrinsic spatial diversity~\cite{wong2020fluid}. Their findings demonstrated that selecting a single port at the FAS could achieve better outage performance compared to fixed-position antenna (FPA) systems. Given their promising performance, FASs have recently emerged as a critical technology for sixth-generation (6G) wireless networks, attracting considerable research interests.


Recent advances~\cite{new2023information,new2023fluid,9715064,10309171} have illustrated the advantages of FASs in enhancing rate and reliability in point-to-point communications. Additionally, FASs have been applied in multi-user communication scenarios. For instance, Wong~\textit{et al.}~\cite{kk_FAMA} explored multi-user FASs by activating the port with the highest signal-to-interference ratio, showcasing the potential of fluid antennas (FAs) to support a large number of users within a compact space of just a few wavelengths. New~\textit{et al.}~\cite{wee2023fluid} focused on maximizing the sum rate for orthogonal multiple access (OMA) and non-orthogonal multiple access (NOMA) through optimal port selection and power allocation. Wu~\textit{et al.}\cite{wu2023movable} proposed a power minimization algorithm for multiple-input single-output (MISO) downlink systems with practical discrete FA positions, while Qin~\textit{et al.}~\cite{sqin2024antenna} examined a FAS MISO system where the base station (BS) and users are equipped with FPAs and FAs, respectively.

In parallel, reconfigurable intelligent surfaces (RISs) have gained significant attention for their ability to construct favorable propagation environments~\cite{8741198,yu2020robust}. By harnessing the reconfigurability of both FAs and RISs, one can jointly optimize the port (antenna) selection of FAs and the phase shifts of RISs to enhance communication performance. However, research on FAS-RIS systems is still in its early stages, with most existing studies focusing on single-port selection~\cite{yao2025fas,lai2024fas}. Consequently, the achievable performance of FAS-RIS systems with multiple selected ports remains unexplored. More importantly, various studies on FASs assume perfect instantaneous channel state information (CSI)~\cite{sqin2024antenna,wu2023movable}, which is challenging to obtain in practical scenarios for two main reasons. First, it is of high complexity to acquire instantaneous CSI. In particular, FASs typically consist of a large number available ports but limited number of RF chains, making it prohibitively costly to acquire instantaneous CSI of all ports for port selection. Moreover, it is challenging to obtain the instantaneous CSI for the base station-RIS and RIS-user links due to the passive nature of RISs~\cite{zappone2021intelligent}. Second, even if perfect CSI can be efficiently obtained, frequently switching the ports of FAs and tuning phase shifts to adapt to instantaneous CSI introduces significant computation and implementation complexity.

To address these challenges, we propose a two-timescale design approach that optimizes the ergodic sum rate instead of instantaneous performance~\cite{han2019large,zhi2022two}. In this framework, port selection and phase shifts are optimized based solely on statistical CSI (e.g., channel correlation and path loss), thereby reducing the frequency of updates for the port selection of FAs and phase shifts of RISs. The concerned FAS-RIS design is different from RIS-aided MIMO systems with antenna selection since existing RIS-aided systems which mainly focus on either point-to-point MIMO systems~\cite{he2021reconfigurable} or multi-user systems with instantaneous channel state information of individual links~\cite{xu2022antenna,abdullah2022low}. To the best of the authors' knowledge, the fundamental limits of FAS-RIS multi-user systems and the associated two-timescale design are not yet available in the literature.


In this paper, we investigate the fundamental limits and two-timescale design for FAS-RIS multi-user systems. Given the regularized zero-forcing (RZF) and zero-forcing (ZF) precoding have been shown to effectively manage the interference between the users in a low-complexity manner~\cite{kammoun2019asymptoticit,kammoun2019asymptotic,wagner2012large}, we will investigate the two-timescale design for FAS-RIS systems with RZF/ZF precoding. Different from existing works~\cite{yao2025fas,lai2024fas} that only considered one-port selection for point-to-point communications, multi-port selection is considered to serve multiple users. Specifically, we maximize the ergodic sum rate (ESR) by jointly optimizing the port selection, the regularization parameter of the RZF precoding, and phase shifts of the RIS. 


\subsection{Challenges}
The two-timescale design of  FASs and RISs with RZF precoding faces several significant challenges, primarily due to the lack of closed-form evaluations for the ESR. While the ESR has been extensively studied in traditional MISO systems~\cite{wagner2012large} and RIS-aided systems without direct links~\cite{zhang2023irs}, the analysis specific to multi-user FAS-RIS systems with direct link and RZF precoding is currently absent from the literature. One of the main difficulties arises from the nature of the end-to-end channel matrix, which is composed of both a random matrix and a product of random matrices. This complexity makes the evaluation of the ESR considerably more intricate than in traditional single-hop MISO systems~\cite{wagner2012large} and two-hop channels~\cite{zhang2022secrecy,zhang2023irs}. Additionally, the inverse operation involved in the RZF/ZF precoding matrix, along with its dependence on the channel vectors, adds further complexity to the analysis. Moreover, the dense deployment of ports in FASs introduces significant correlation among them, which cannot be overlooked. This correlation is crucial for effective port selection and exhibits complex and varied structures~\cite{ramirez2024new,wong2022closed}, complicating the evaluation of the fundamental limits of FAS-RIS systems. Finally, the implicit nature of the ESR concerning the port selection vector, regularization factor, and phase shifts presents additional challenges for the two-timescale design.

\subsection{Contributions}

The main contributions of this paper are summarized as follows.

\textit{1. Closed-Form Evaluation of SINR and ESR}: We derive a closed-form evaluation for the signal-to-interference-plus-noise ratio (SINR) and per-user ergodic rate of FAS-RIS MISO systems employing RZF and ZF precoding, while considering the direct link between the BS and users. Our results accommodate arbitrary antenna correlation matrices for both FASs and RISs, thereby facilitating practical FAS design. We analyze two system scenarios: uncommon correlation, where different users have distinct correlation matrices at the FAs, and common correlation, where users share a common correlation matrix at the RIS. The derived results can degenerate to ESR evaluations over single-hop and two-hop channels. Notably, the proposed approach establishes a general framework for evaluating the ESR of RIS-aided MISO broadcasting systems with direct links, a topic that has not been addressed in the existing literature.

\textit{2. Physical Insights into ESR Performance}: We provide valuable physical insights by deriving the ESR with ZF precoding over independent and identically distributed (i.i.d.) channels. We explicitly demonstrate how several key system parameters—such as the number of selected ports, users, and reflecting elements at the RIS. Our findings indicate that the FAS-RIS system requires fewer selected ports than traditional FAS systems to achieve the same ESR. Additionally, we show that the ESR with ZF precoding outperforms that achieved with maximal ratio transmission (MRT) precoding.

\textit{3. Two-Timescale Algorithm for Optimization}: Building on the derived ESR, we propose a two-timescale algorithm designed to maximize rate performance by jointly optimizing port selection, the regularization parameter of RZF precoding, and the phase shifts at the RIS. To address the discrete optimization challenge associated with port selection, we employ linear relaxation and the Frank-Wolfe (FW) method, deriving a closed-form update rule that fully exploits the structure of the port selection constraint. This approach is applicable to FAS-RIS systems with arbitrary correlation matrices.

\textit{4. Validation through Numerical Results}: We validate the accuracy of our ESR evaluation and the effectiveness of the proposed algorithm through numerical simulations. Our results demonstrate that the impact of joint FAS and RIS optimization becomes more pronounced in FAS-RIS systems with a larger number of users. Furthermore, we establish that for FAS-RIS systems with homogeneous channels, the optimal regularization factor is independent of port selection, correlation matrices, and phase shifts at the RIS.

\subsection{Paper Outline and Notations}
The rest of the paper is organized as follows. Section~\ref{sec_mod} introduces the multi-user FAS-RIS system and problem formulation. Section~\ref{sec_res} determines the ESR of the concerned system and Section~\ref{large_analysis} evaluates the performance over the i.i.d. channels and analyzes the impact of the number of reflecting elements at the RIS and number of selected ports. Section~\ref{sec_algorithm} gives the two-timescale design for ESR maximization and Section~\ref{sec_simu} shows numerical simulations. Finally, Section~\ref{sec_con} concludes the paper. 

\textit{Notations:} Vectors and matrices are denoted by bold lower case letters and bold upper case letters, respectively. The $N$-dimensional vector and $M$-by-$N$ matrix space are represented by $\mathbb{C}^{N}$ and $\mathbb{C}^{M\times N}$, respectively. The $(i,j)$-th entry, conjugate transpose, spectral norm, and trace of matrix $\bold{A}$ are represented by $[\BA]_{i,j} (A_{i,j})$, $\bold{A}^{H}$, $\|\BA \|$, and $\Tr(\BA)$, respectively. $\bold{I}_{N}$ represents the $N$-by-$N$ identity matrix. The expectation and centered form of a random variable $x$ is denoted by $\E[ x]$ and $\underline{x}=x-\E [x]$. The covariance of $x$ and $y$ is represented by $\cov(x,y)=\E[\underline{x}\underline{y}] $ and the convergence in probability is denoted by $\xrightarrow[N \rightarrow \infty]{\mathcal{P}}$. The big-O and little-o notations are represented by $\BO(\cdot)$ and $o(\cdot)$. $\mathbbm{1}$ represents the indicator function and $\jmath=\sqrt{-1}$. The complex partial derivative operators with respect to $z=x+\jmath y$  are given by $\frac{\partial}{\partial z}=\frac{1}{2}(\frac{\partial}{\partial x}-\jmath \frac{\partial}{\partial y})$ and $\frac{\partial}{\partial {z}^{*}}=\frac{1}{2}(\frac{\partial}{\partial x}+\jmath \frac{\partial}{\partial y})$, where $\frac{\partial}{\partial x}$ and $\frac{\partial}{\partial y}$ are real partial derivatives with respect to $x$ and $y$, respectively.

\section{System Model and Problem Formulation}
\label{sec_mod}
In this section, we first introduce the FAS-RIS multiuser downlink MISO systems. Then, we formulate the problem for the two-time-scale design by maximizing the ESR.

\subsection{System Model}
Consider a FAS-RIS downlink MISO system, where a BS equipped with $M$ FAs and $M_{\mathrm{tot}}$ ports serves $K$ single-antenna users with the aid of a RIS with $L$ reflecting elements. Here, a ``port'' represents a possible position of the antenna~\cite{wong2020fluid} and each FA has one RF chain. Different from traditional antenna selection systems with $M_\mathrm{tot}$ antennas that can be turned on and off and connected to $M$ RF chains, the FAS consists of only the $M$ FAs that can switch the location of its radiating element to one of the $M_{\mathrm{tot}}$ preset locations/ports~\cite{lai2023performance,9715064}. Under such circumstances, the BS could dynamically select $M$ ports from the total $M_{\mathrm{tot}}$ available ports to create favorable propagation environments. Here, we adopt a planar FAS with aperture $W_x\times W_y$ consisting of $M_{\mathrm{tot}}=N_x\times N_y$ uniformly deployed ports, where $W_x$ and $W_y$ are the length and width of the array measured in wavelength, and $N_x$ and $N_y$ are the number of antennas in each row and column, respectively. Note that there exists the direct link between the BS and users, and the RIS is deployed to enhance the communication link. It is worth noticing that there are fundamental operational differences between FAS and traditional systems with antenna selection. Traditional antenna selection exploiting multiple fixed antennas tend to have antenna spacing of at least $\frac{\lambda}{2}$ to achieve spatial diversity while the ports of a fluid antenna in the FAS should be as close as possible to ensure sufficient spatial resolution to create a favorable correlation structure to improve performance~\cite{wong2022fast}.


The port selection scheme can be represented by a selection vector $\bold{s}\in \{0,1 \}^{M_{\mathrm{tot}}}$, where $s_i=1$ represents that the $i$-th port is selected and $\bold{s}$ satisfies $\bold{s}^{T}\bold{1}_{M_{\mathrm{tot}}}=M$. Let $\bold{m}=[m_{1},m_{2},...,m_{K}]^{T}\in\mathbb{C}^{K}$ denote the transmit signal where $m_k\sim\mathcal{CN}(0,p_k)$ represents the message for the $k$-th user with transmit power $p_k$. The received signal of the $k$-th user is given by~\cite{kammoun2019asymptotic}
\begin{equation}
\begin{aligned}
y_{k}=\bold{h}_{k}^{H}(\bold{s})(\bold{g}_{k} m_{k}+\sum_{i\neq k}\bold{g}_{i}m_{i})+{n}_{k},
\end{aligned}
\end{equation}
where $\Bh_k(\Bs) \in \mathbb{C}^{M} $ and $\bold{g}_{k}\in \mathbb{C}^{M}$ represent the channel vector for the selected ports and precoding vector for the $k$-th user, respectively, and $n_k\sim \mathcal{CN}(0,\sigma^2)$ denotes the additive white Gaussian noise (AWGN). Here $\Bh_k(\Bs)$ can be obtained by sampling from $\Bh_{\mathrm{tot},k}$ according to $\Bs$, where $\Bh_{\mathrm{tot},k}\in \mathbb{C}^{M_{\mathrm{tot}}}$ represents the channel vector from all ports to the $k$-th user. In the following, we omit $(\Bs)$ for brevity. The signal to interference plus noise ratio (SINR) of the $k$-th user is given by
\begin{equation}
\label{sinr_def}
\gamma_{k}=\frac{p_k|\bold{h}_{k}^{H}\bold{g}_{k}|^2 }{\sum\limits_{i\neq k} p_i |\bold{h}_{k}^{H}\bold{g}_{i}|^2+\sigma^2},
\end{equation}
such that the ESR can be represented as
\begin{equation}
R= \sum_{k=1}^{K} \E[  \log(1+\gamma_k)].
\end{equation}
Designing the optimal transmitter for the concerned FAS-RIS system is extremely difficult due to the vacancy of closed-form evaluation for the ESR. In this paper, we will investigate the closed-form evaluation for the ESR. Given RZF and ZF precoding have been shown effective in managing the interference among the users in a low-complexity manner, we consider the design of FAS-RIS systems with RZF/ZF precoding.




\subsection{Linear Precoding}
The precoding matrix for RZF and ZF can be given by~\cite{kammoun2019asymptotic,wagner2012large} 
\begin{align}
\BG_{\mathrm{RZF}}&=\frac{1}{\sqrt{\xi_{\mathrm{RZF}}}}\left({\BH}{\BH}^{H}+  z\bold{I}_{M} \right)^{-1}{\BH}, \nonumber
\\
\BG_{\mathrm{ZF}}&=\frac{1}{\sqrt{\xi_{\mathrm{ZF}}}}{\BH}\left({\BH}^{H}{\BH} \right)^{-1},
\end{align}
respectively, where $z>0$ is the regularization parameter of RZF. The power normalization factor $\xi$ satisfies $ \Tr(\BG\BP\BG^{H})=1$ with $\BP=\diag(p_1,p_2,...,p_K)$. In particular, we have $\xi_{\mathrm{RZF}}=\frac{1}{M}\Tr\left(\left(\BH\BH^{H}+ z\bold{I}_{M} \right)^{-2}\BH\BP\BH^{H}\right)$ for RZF and $\xi_{\mathrm{ZF}}=\frac{1}{M}\Tr\left(\left(\BH^{H}\BH \right)^{-1}\BP\right)$ for ZF, respectively.

Note that the channel matrix $\BH$ is related to the port selection. The channel matrix of all available ports can be given by $\BH_{\mathrm{tot}}\in \mathbb{C}^{M_{\mathrm{tot}}\times K}=[\Bh_{\mathrm{tot},1},\Bh_{\mathrm{tot},2},...,\Bh_{\mathrm{tot},K}]$. For RZF precoding, we have $\Bh_{k}^{H}\BG=\Bh_{\mathrm{tot},k}^{H}\diag(\Bs)(z\BI_{M_{\mathrm{tot}}}
+\diag(\Bs)\BH_{\mathrm{tot}}\BH^{H}_{\mathrm{tot}}\diag(\Bs) )^{-1} \diag(\Bs)\BH_{\mathrm{tot}}$
which indicates that $\BH$ can be replaced by $\diag(\Bs)\BH_{\mathrm{tot}}$ in the ESR evaluation. The same replacement holds true for ZF precoding.

\subsection{Channel Model}
\label{cha_mol}
We consider the correlated Rayleigh fading channels, which could characterize both the random scattering and correlation among reflecting ports/elements in FAS-RIS systems~\cite{wong2020performance,papazafeiropoulos2016impact}. Furthermore, we consider both uncommon~\cite{wagner2012large} and common correlation scenarios~\cite{kammoun2019asymptoticit}. 

\subsubsection{FAS-RIS with Uncommon Correlation} 
With uncommon correlation, the BS and RIS have different correlation matrices towards different users, which happens when the angle from each user to RISs and FASs is significantly different from each other~\cite{wagner2012large,papazafeiropoulos2016impact}. The channel vector of the $k$-th user is given by
\begin{equation}
\label{chunc}
\Bh_k=\sqrt{u_k}\LF_{c,k}(\Bs)\Bw_k+ \sqrt{t_k} \LR(\Bs)\BX\MC_{L} \bold{\Phi}\MC_{R,k}\By_k,
\end{equation}
where $\BF_{c,k}(\Bs)\in \mathbb{C}^{M\times M}$ and $\FR(\Bs)\in \mathbb{C}^{M\times M}$ represent the correlation matrix of selected ports at the BS towards the $k$-th user and the RIS, respectively, and $\Bs\in \{0,1 \}^{M_{\mathrm{tot}}}$ denotes the port selection vector. Here $ \FC_{L} \in \mathbb{C}^{L\times L}$ and $\FC_{R} \in \mathbb{C}^{L\times L}$ represent the correlation matrices at the receive and transmit side of the RIS, respectively, and $\bold{\Phi}\in \mathbb{C}^{L\times L}=\diag(e^{\jmath \phi_1},e^{\jmath \phi_2},...,e^{\jmath \phi_L})$ denotes the phase shift matrix of the RIS with $\phi_l\in [0,2\pi]$, $l=1,2,...,L$.  Note that in FAS-RIS systems, the transmit correlation matrix is determined by antenna correlation and changes as $\Bs$ varies. The large-scale channel gains of the direct and cascaded link of the $k$-th user are denoted by $u_k$ and $t_k$, respectively. Matrix $\BX \in  \mathbb{C}^{M\times L} $ is a random matrix with i.i.d. entries following $ \mathcal{CN}(0,\frac{1}{M})$, which models the small-scale fading of the BS-RIS hop. The i.i.d. vectors $\Bw_k \in \mathbb{C}^{M} $ with $w_1\sim \mathcal{CN}(0,\frac{1}{M})$ and $\By_k \in \mathbb{C}^{L} $ with $y_1\sim \mathcal{CN}(0,\frac{1}{L})$ represents the small-scale fading of the direct link and RIS-user hop, respectively. For brevity, we denote $\FF_k(\Bs)=u_k\BF_{c,k}(\Bs)$ and $\FC_i=\LC_i (\LC_i )^{H}$ with $\LC_i=\sqrt{t_i}\MC_{L} \bold{\Phi}\MC_{R,i}$ and omit $(\Bs)$ such that $\FF_k(\Bs)=\FF_k$ and $\FR(\Bs)=\FR$. Given $\BF_{\mathrm{tot},k}\in \mathbb{C}^{M_{\mathrm{tot}}\times M_{\mathrm{tot}}}$ and $\FR_{\mathrm{tot}}\in \mathbb{C}^{M_{\mathrm{tot}}\times M_{\mathrm{tot}}}$ are correlation matrices of all available ports, $\BF_{k}$ and $\FR$ are the $M$-by-$M$ sub-matrices sampled by $\Bs$ from $\BF_{\mathrm{tot},k}$ and $\FR_{\mathrm{tot}}$, respectively.

\subsubsection{FAS-RIS with Common Correlation}
The common correlation scenario can be utilized to describe an environment where multiple users are heterogeneously scattered around the transmitter, e.g., users are located in a compact direction from the RIS and FAS but experience different channel gains, such that all user channels have identical transmit correlation. This assumption has been widely in the tractable analysis of MIMO and RIS systems~\cite{yang2014confidential,kammoun2020asymptotic,zhang_tifs}. Under such circumstances, the BS and RIS have common correlation matrix towards different users, but the large-scale channel gain of different users is different since the distances from users to the BS and RIS are different. Thus, the channel vector for the $k$-th user can be represented by 
\begin{equation}
\label{chcom}
\Bh_k= \sqrt{u_k}\LF\Bw_k+ \sqrt{t_k}\LR\BX\MC_{L} \bold{\Phi}\MC_{R}\By_k.
\end{equation}
It can be observed that the the common correlation scenario is a special case of the uncommon correlation scenario when $\FF_{k}=u_k\FF$ and $\FC_{R,k}= \FC_R$. This corresponds to the cases where the users are located in a narrow angular direction from the RIS~\cite{yang2014confidential}. The channel matrix $\BH=[\bold{h}_{1},...,\bold{h}_{K}]$ can be rewritten as $\BH=\LF\BW\RU+\LR\BX\MC_{L} \bold{\Phi}\MC_{R}\BY\RT,
$ with $\FT=\diag(t_1,t_2,...,t_K)$ and $\BU=\diag({u}_{1},u_2,...,{u}_{K})$. For brevity, we denote $\FC=\LC(\LC)^{H}$ with $\LC=\MC_{L} \bold{\Phi}\MC_{R}$.

\subsection{Problem Formulation}
In this paper, we aim to maximize the ESR of the FAS-RIS-aided MISO system $\E [R(\Bs,\bold{\Phi},z)]$, by jointly optimizing the port selection vector $\Bs$, regularization factor $z$ of RZF, and the phase shifts of the RIS $\bm{\Phi}$. The corresponding optimization problem is formulated as
\begin{equation}
\label{p1_exp}
\begin{aligned}
\mathcal{P}1:~& \max_{z>0,\bold{\Phi},\Bs} \E [R(\Bs,\bold{\Phi},z)],
\\
\mathrm{s.t.}~& \mathcal{C}_1:\Bs^{T}\bold{1}_{M_{\mathrm{tot}}}=M,~\Bs\in\{0,1\}^{M_{\mathrm{tot}}},\\
&\mathcal{C}_2:\bold{\Phi}=\diag(e^{\jmath \phi_1},e^{\jmath \phi_2},...,e^{\jmath \phi_L}).
\end{aligned}
\end{equation}
In the next section, we first evaluate the ESR for concerned FAS-RIS systems.


\section{Evaluation of Ergodic Sum Rate}
\label{sec_res}

In this section, we first derive a closed-form approximation for the ESR with RZF for both the uncommon and common correlation scenarios. Then, we extend the results to obtain an approximation for the ESR with ZF. Define $c_1=\frac{K}{M}$ and $c_2=\frac{K}{L}$. We take the following assumptions in this work, which have been widely adopted in the literature.   

\textbf{A.1.} $0<\liminf\limits_{K \ge 1}  c_1 \le c_1  \le \limsup\limits_{K \ge 1} c_1 <\infty$ and $0<\liminf\limits_{K \ge 1}  c_2 \le c_2  \le \limsup\limits_{K \ge 1}  c_2 <\infty$.

\textbf{A.2.} $\limsup\limits_{K\ge 1}  \{ \| \FR \|, \|\BF_k \|,\|\FC_k\|,\|\BC\|,\|\BF\|,\|\BU\|,\|\FT\|  \} <K_{\mathrm{op}}$, where $K_{\mathrm{op}}$ is a constant independent of $K$, $M$, and $L$.

\textbf{A.3.} $\inf\limits_{M\ge 1}\{ \frac{1}{M}\Tr(\FR),\frac{1}{M}\Tr(\BF_k),\frac{1}{M}\Tr(\BC_k)  \} >0  $.

 \textbf{A.1} indicates that the dimensions of the system ($M$, $K$, and $L$) grow to infinity proportionally with ratios $c_1$ and $c_2$, and the asymptotic regime is denoted by $K \xlongrightarrow{(c_1,c_2)} \infty$. \textbf{A.2} and~\textbf{A.3} guarantee that the channel rank is not extremely low, i.e., the rank of correlation matrices increases with the number of antennas increases~\cite{hachem2008new,zhang2021bias}. In the following, we will present the closed-form evaluation of the ESR with both uncommon and common cases.
 
 By utilizing RMT, the fundamental limits of the considered system can be characterized by a system of fixed-point equations parameterized by statistical CSI. Such an approach has been widely used in the ergodic capacity and outage probability analysis of large-scale MIMO systems~\cite{kk_rmt,wen2012deterministic,wagner2012large,zhang2013capacity} and RIS-aided systems~\cite{zhang2022outage,zhang2023asymptotic,zhang2022secrecy,zhang2021large}. In the following, we leverage the same method to evaluate the ESR for both uncommon and common correlation cases.

\subsection{Uncommon Correlation Case}

\subsubsection{ESR with RZF}
The key parameters $\delta(z)$, $\mu_{k}(z)$, and $\omega_k(z)$, $k=1,2,...,K$ for the closed-form ESR evaluation are determined by the following system of equations
 \begin{equation}
 \label{basic_eq1}
 \begin{aligned}
 \begin{cases}
   \delta(z)&=\frac{1}{M}\Tr(\bold{R}\BPS_R),
   \\
   \omega_k(z)&=\frac{1}{L}\Tr(\BC_k\BPS_{C}),
   \\
    \mu_k(z) &=\frac{1}{M}\Tr(\FF_k\BPS_R)+ \omega_k(z),
\end{cases}
  \end{aligned}
 \end{equation}
 where $\BPS_{R}=\left(z \bold{I}_{M}  + \sum_{k=1}^{K} \frac{\FF_k}{M(1+\mu_k(z))}+ \frac{ \omega_k(z)\FR}{M\delta(z)(1+\mu_k(z))}  \right)^{-1},
$ and $\BPS_{C}  =\left( \frac{1}{\delta(z)}\bold{I}_{L}+\sum_{k=1}^{K}\frac{\FC_k}{L(1+\mu_k(z))}   \right)^{-1}.
$. The solution of~(\ref{basic_eq1}) can be obtained by the fixed-point algorithm~\cite{zhang2021bias}. By following similar analysis in~\cite[Appendix I.C]{wagner2012large}, we can obtain that the complexity of obtaining an $\varepsilon$-approximate solution is $\BO(M\log^3(\varepsilon^{-1}))$. For ease of illustration, we will omit $z$ in $(\delta(z),\mu_1(z),\mu_{2}(z),...,\mu_K(z))$ and use $(\delta,\bm{\mu})=(\delta,\mu_1,\mu_{2},...,\mu_K)$ in the following. With these notations, we obtain the ESR evaluation shown in the following theorem. 


 \begin{theorem} (SINR and ESR with RZF)
\label{uncommon_the}
Under the assumptions~\textbf{A.1}-\textbf{A.3} and denoting by $(\delta,\bm{\mu})$ the positive solution of~(\ref{basic_eq1}), the SINR $\gamma_{\mathrm{RZF},k}$ with regularization factor $z$ satisfies
\begin{equation}
\gamma_{\mathrm{RZF},k}  \xlongrightarrow[K \xlongrightarrow{(c_1,c_2)} \infty]{\mathcal{P}} \overline{\gamma}_{\mathrm{RZF},k},
\end{equation}
 where 
\begin{equation}
\label{SINR_de}
\begin{aligned}
\overline{\gamma}_{\mathrm{RZF},k} &=\frac{p_k \mu_k^2}{ \sum\limits_{l\neq k} \frac{p_l\Psi_{k,l}}{L(1+\mu_l)^2}  + \sigma^2(1+\mu_k)^2\overline{C} },
\end{aligned}
\end{equation}
$\Psi_{k,l}$ and $\overline{C}$ are given in~(\ref{uncommon_parexp_1}) and~(\ref{uncommon_parexp_2}) at the top of the next page.
\begin{figure*}
\begin{equation}
\label{uncommon_parexp_1}
\Psi_{k,l}=[\bm{\Delta}^{-1}\Bxi_l]_k+\frac{L}{M\delta^2}[\bm{\Delta}^{-1}\Bxi_I]_k [\bm{\Pi}^{-1}\bm{\chi}(\FR)]_l   
+\frac{L}{M}[\bm{\Pi}^{-1}\bm{\chi}(\overline{\BF}_k)]_{l}+\frac{L}{M} [\bm{\Pi}^{-1}\bm{\chi}(\BF_k)]_l,
\vspace{-0.3cm}
\end{equation}
\begin{minipage}{0.59\textwidth}
\begin{equation}
\begin{aligned}
\bm{\chi}(\BK) &=[\chi(\BF_1,\BK),\chi(\BF_2,\BK),...,\chi(\BF_K,\BK),\chi(\FR,\BK) ]^{T},\\
\bm{\xi}_l &= [\Xi_{1,l},\Xi_{2,l},...,\Xi_{K,l}]^{T},~\bm{\xi}_I=[\Xi_{I,1},\Xi_{I,2},...,\Xi_{I,L}]^{T},\\
[\bm{\Pi}]_{k,l}&=
\begin{cases}
\mathbbm{1}_{l=k}-\frac{\Xi_{k,l}}{L(1+\mu_l)^2}-\frac{ (\frac{\Xi_{I,l}}{\delta^2} \chi(\BF_k,\FR)+\chi(\BF_k,\BF_l))}{M(1+\mu_l)^2},~1\le k,l \le K,
\\
-\frac{ (\frac{\Xi_{I,l}}{\delta^2} \chi(\FR,\FR)+\chi(\FR,\BF_l))}{M(1+\mu_l)^2}, ~k=K+1, l\le K,
\\
-\frac{\Xi_{I,k}}{\delta^2}-\sum\limits_{m=1}^{K}\frac{(\omega_m-\frac{\Xi_{I,m}}{\delta})}{M\delta^2(1+\mu_m)} \chi(\BF_k,\FR), ~k\le K, l=K+1,
\\
1-\sum\limits_{m=1}^{K}\frac{(\omega_m-\frac{\Xi_{I,m}}{\delta})}{M\delta^2(1+\mu_m)} \chi(\FR,\FR),~~k= K+1, l=K+1,
\end{cases}
\end{aligned}
\nonumber
\end{equation}
\end{minipage}
\begin{minipage}{0.41\textwidth}
\begin{align*}
[\bm{\Delta}]_{k,l}&= \mathbbm{1}_{l=k}-\frac{\Xi_{k,l}}{L},~\Xi_{I,k}=\frac{1}{L}\Tr(\BC_k\BPS_C^2),
\\
\Xi_{k,l} &=\frac{1}{L}\Tr(\BC_k\BPS_C\BC_l\BPS_C),
\\
\chi(\BA,\BB)& =\frac{1}{M}\Tr(\BA\BPS_R\BB\BPS_R),
\\
\overline{C}&=\sum\limits_{m=1}^{K}\frac{p_m [\bm{\Pi}^{-1}\bm{\chi}(\BI_M)]_{m} }{M(1+\mu_m)^2},\numberthis \label{uncommon_parexp_2}
\\
\overline{\BF}_k& =\sum_{l=1}^{K}[\bm{\Delta}^{-1}(\BI_K-\bm{\Delta})]_{k,l}\BF_l.
\end{align*}
\end{minipage}
\hrule
\vspace{-0.2cm}
\end{figure*}
Furthermore, it holds true
\begin{equation}
\begin{aligned}
\label{eva_uncommon}
\frac{R_{\mathrm{RZF}}}{K}   \xlongrightarrow[K \xlongrightarrow{(c_1,c_2)} \infty]{\mathcal{P}}  \frac{\overline{R}_{\mathrm{RZF}}}{K},
\frac{\E[R_{\mathrm{RZF}}]}{K}   \xlongrightarrow[]{K \xlongrightarrow{(c_1,c_2)} \infty}   \frac{\overline{R}_{\mathrm{RZF}}}{K},
\end{aligned}
\end{equation}
where $\overline{R}_{\mathrm{RZF}}=\sum_{k=1}^{K}\log(1+\overline{\gamma}_{\mathrm{RZF},k} )$.

\end{theorem}
\begin{proof} The proof of Theorem~\ref{uncommon_the} is given in Appendix~\ref{uncommon_the_proof}.
\end{proof}
\begin{remark} Theorem~\ref{uncommon_the} indicates that when the number of selected ports, the number of reconfigurable elements at the RIS, and number of users go to infinity proportionally, the SINR and the rate of the $k$-th user converges to a deterministic value. This provides a good approximation for the SINR and ESR. It is worth noticing that Theorem~\ref{uncommon_the} is applicable for general correlation structure and is also valid when replacing $\BF_k$ and $\FR$ with $\FF^{\frac{1}{2}}_{\mathrm{tot},k}\diag(\Bs)\FF^{\frac{1}{2}}_{\mathrm{tot},k}$ and $\FR^{\frac{1}{2}}_{\mathrm{tot}} \diag(\Bs)\FR^{\frac{1}{2}}_{\mathrm{tot}}$, respectively, for optimization purposes. The evaluation result is also applicable for the case with mutual coupling (MC), where the impedance matrix is multiplied with the correlation matrices at the BS~\cite{sun2011capacity}. We do not discuss MC effect here since the MC model for FASs is not available now.
\end{remark}

In Theorem~\ref{uncommon_the}, we consider the RIS channel with the direct link, whose channel matrix is represented by the sum of a random matrix and product of random matrices. Thus, the ESR evaluation results could degenerate to those with the single-hop channels~\cite{wagner2012large} and two-hop channels~\cite{zhang2023irs} as shown in Remarks~\ref{rem_deg_sig} and~\ref{rem_deg_two}, respectively.

\begin{remark}\label{rem_deg_sig} (\textbf{Degeneration to single-hop channels}) The SINR evaluation for correlated single-hop channels was investigated in~\cite{wagner2012large}. By letting $\FR=\bm{0}_{M\times M}$, we can obtain $\omega_k=0$ and $\delta=0$ such that  the system of equations in~(\ref{basic_eq1}) degenerates to that in~\cite[Eq. (11)]{wagner2012large}. Thus, the associated SINR evaluation in (\ref{SINR_de}) can degenerate to that in~\cite[Eq. (19) with $\tau_k=0$]{wagner2012large}.
\end{remark}

\begin{remark}\label{rem_deg_two} (\textbf{Degeneration to two-hop channels}) The SINR evaluation for correlated two-hop channels was investigated in~\cite{zhang2023irs,zhang_tifs}. By letting $\FF_k=\bm{0}_{M\times M}$, $k=1,2,...,K$, we can obtain $\mu_k=\omega_k$ such that the system of equation in~(\ref{basic_eq1}) degenerates to that in~\cite[Eq. (11)]{zhang2023irs}. Thus, the associated SINR evaluation in (\ref{SINR_de}) can degenerate to that in~\cite[Eq. (23)]{zhang2023irs} and~\cite[Eq. (12)]{zhang_tifs}.
\end{remark}

\subsubsection{ESR with ZF} 
The ESR with ZF can be given by the following theorem. 
\begin{theorem}\label{the_ZF} (SINR and ESR with ZF) If it holds true that $ \lambda_{\mathrm{min}} (\frac{1}{M} \BH^{H}\BH) >\varepsilon $ with probability $1$ for an $\varepsilon>0$, and $\lim\limits_{z\rightarrow 0}z\delta(z)>0$, $\lim\limits_{z\rightarrow 0}z\omega_i(z)>0$, and $\lim\limits_{z\rightarrow 0}z\mu_i(z)>0$ for $i=1,2,...,K$ exist, the SINR with ZF can be approximated by
\begin{equation}
{\gamma}_{\mathrm{ZF},k}  \xlongrightarrow[K \xlongrightarrow{(c_1,c_2)} \infty]{\mathcal{P}} \overline{\gamma}_{\mathrm{ZF},k},  
\end{equation}
where $ \overline{\gamma}_{\mathrm{ZF},k}= p_k(\sigma^2\sum\limits_{l=1}^{K}\frac{p_l}{M\underline{\mu}_l})^{-1} $ and $\underline{\mu}_k$ is the solution of the following system of equations
\begin{equation}
 \label{basic_zf1}
 \begin{aligned}
 \begin{cases}
   &\underline{\delta}=\frac{1}{M}\Tr(\bold{R}\BK_R),
   \\
 &\underline{\mu}_k=\underline{\omega}_k+\frac{1}{M}\Tr(\BF_k\BK_R),
 \\
 &\underline{\omega}_k=\frac{1}{L}\Tr( \FC_k\BK_{C}),
\end{cases}
  \end{aligned}
 \end{equation}
$k=1,2,...,K$, where $\BK_{R}=\left( \bold{I}_{M} + \sum_{k=1}^{K}\frac{\underline{\omega}_k\FR}{M\underline{\delta}\underline{\mu}_k}+ \frac{\BF_k}{M\underline{\mu}_k}  \right)^{-1}$ and $\BK_{C}=\left( \frac{1}{\underline{\delta}}\bold{I}_{L}+\sum_{k=1}^{K}\frac{\FC_k}{L\underline{\mu}_k}   \right)^{-1}$.
Furthermore, it holds true
\begin{equation}
\begin{aligned}
\label{eva_uncommon}
\frac{R_{\mathrm{ZF}}}{K}   \xlongrightarrow[K \xlongrightarrow{(c_1,c_2)} \infty]{\mathcal{P}}  \frac{\overline{R}_{\mathrm{ZF}}}{K},
\frac{R_{\mathrm{ZF}}}{K}   \xlongrightarrow[]{K \xlongrightarrow{(c_1,c_2)} \infty}   \frac{\overline{R}_{\mathrm{ZF}}}{K},
\end{aligned}
\end{equation}
where $\overline{R}_{\mathrm{ZF}}=\sum_{k=1}^{K} \log(1+ \overline{\gamma}_{\mathrm{ZF},k})$.
\end{theorem}
\begin{proof} The proof of Theorem~\ref{the_ZF} is given in Appendix~\ref{proof_the_ZF}.
\end{proof}
\begin{remark} The assumption on $\lambda_{\mathrm{min}} (\frac{1}{M} \BH^{H}\BH)$ is to guarantee the feasibility of ZF precoding. In practical systems, when the number of selected ports is larger or equal to the number of users, i.e., $M \ge K$, ZF precoding is feasible. Compared with ZF, RZF does not have the limitation and is applicable for any $M$ and $K$. 
\end{remark}


\subsection{Common Correlation Case}
 \subsubsection{ESR with RZF}
For common correlation scenario, the ESR with RZF can be characterized by five key parameters $\delta$, $\kappa$, $\overline{\kappa}$, $\omega$, and $\overline{\omega}$, which are determined by the following system of equations.
 \begin{equation}
 \label{basic_eq2}
 \begin{aligned}
 \begin{cases}
   &\delta=\frac{1}{M}\Tr(\bold{R}\BPS_R),
   \\
   & \kappa=\frac{1}{M}\Tr(\FF\BPS_R),
   \\
 &\omega=\frac{1}{L}\Tr( \FC\BPS_{C}),
 \\
  &  \overline{\kappa}=\frac{1}{L}\Tr(\BU\BPS_{T}),
\\
 & \overline{\omega}=\frac{1}{L}\Tr(\FT\BPS_{T}),
\end{cases}
  \end{aligned}
 \end{equation}
 where $\BPS_{R}=\left(z \bold{I}_{M}+\frac{L \omega\overline{\omega}}{M\delta}\bold{R} + \frac{L\overline{\kappa}}{M}\BF  \right)^{-1}$, $ \BPS_{C}=\left(\frac{1}{\delta}\bold{I}_{L}+\overline{\omega}\BC\right)^{-1}$, and $\BPS_{T} =\left(\bold{I}_{K}+\omega\bold{T}+\kappa\BU\right)^{-1}$. The system of equation can also be solved by the fixed-point algorithm and the complexity of obtaining an $\varepsilon$-approximate solution is $\BO(\log^3(\varepsilon^{-1}))$.


 \begin{table*}[!htbp]
\centering
\caption{Symbol Table for FAS-RIS with Common Correlation.}
\label{var_list}
\begin{tabular}{|cc|cc|cc|cc|}
\toprule
Symbol& Expression &  Symbol & Expression &Symbol& Expression \\
\midrule
$\chi(\BA,\BB)$ & $\frac{1}{M}\Tr(\BA\BPS_{R}\BB\BPS_R)$
&
$\Xi$ & $\frac{1}{L}\Tr(\BC^2\BPS_{C}^{2})$
&
 $\Upsilon(\BA)$& $\frac{L\omega}{M\delta}\chi (\FR,\BA)\eta(\FT,\BU)+\frac{L}{M}\chi(\BF,\BA)  \eta(\BU,\BU)$
\\
$\eta(\BA,\BB)$& $\frac{1}{L}\Tr(\BA\BPS_{T}\BB\BPS_T)$
&
$\Xi_I$ & $\frac{1}{L}\Tr(\BC\BPS_{C}^{2})$
 &
  $\Lambda(\BA)$& $\frac{L}{M}\chi(\BF,\BA)\eta(\BU,\BT)-\frac{L}{M\delta} \chi(\FR,\BA)( \overline{\omega}-\omega\eta(\BT,\BT)  )$
  \\
  $\Delta$ & $1-\Xi\eta(\FT,\FT)$ 
   & $\bm{\chi}(\BM)$  & $[ \chi(\FR,\BM),\chi(\FF,\BM),0 ]^{T} $
  & &
\\
\bottomrule
\end{tabular}
\vspace{-0.3cm}
\end{table*}
The SINR and ESR with RZF for the common correlation scenario is given in the following proposition, which can be obtained by setting $\FC_k=t_k\FC$ in Theorem~\ref{uncommon_the} .
\begin{proposition} (SINR and ESR with RZF)
\label{the_ESR_common}
Under the assumptions~\textbf{A.1}-\textbf{A.3} and denoting by $(\delta,\kappa, \overline{\kappa}, \omega, \overline{\omega})$ the positive solution of~(\ref{basic_eq1}), the 
SINR can be approximated by
\begin{equation}
\label{SINR_com_exp}
\gamma_{\mathrm{RZF},k}   \xlongrightarrow[N \xlongrightarrow{(c_1,c_2)} \infty]{\mathcal{P}} \overline{\gamma}_{\mathrm{RZF},k},
\end{equation}
where 
\begin{equation}
\begin{aligned}
\overline{\gamma}_{\mathrm{RZF},k} =\frac{p_k (t_k\omega+u_k\kappa)^2}{\sum\limits_{l \neq k }\frac{p_l  \Psi_{k,l}}{L(1+t_l\omega+u_l\kappa)^2} + (1+t_k\omega+u_k\kappa )^2 \overline{C}_{\mathrm{com}}},
\end{aligned}
\end{equation}
with $\Psi_{k,l}$ and $\overline{C}_{\mathrm{com}}$ given in~(\ref{Psi_exp}) and Table~\ref{var_list} at the top of the next page. Furthermore, it holds true
\begin{figure*}
\begin{align*}
\Psi_{k,l} &= \frac{t_l t_k}{\Delta}(\Xi+\frac{L\Xi_I }{M\delta^2} [\bm{\Pi}_{\mathrm{com}}^{-1}\bm{\chi}(\FR)]_{3} )+(\frac{ Lt_l t_k\Xi\eta(\FT,\BU)}{M\Delta}+\frac{L (t_l u_k+t_k u_l)}{M}) [\bm{\Pi}_{\mathrm{com}}^{-1}\bm{\chi}(\BF)]_{3}+\frac{u_l u_k L}{M} [\bm{\Pi}_{\mathrm{com}}^{-1}\bm{\chi}(\BF)]_{2} ), \numberthis \label{Psi_exp}
\\
\overline{C}_{\mathrm{com}}&=\frac{L}{M}[\eta(\BP,\FT)[\bm{\Pi}_{\mathrm{com}}^{-1}\bm{\chi}(\BI_M)]_{3}+\eta(\BP,\BU)[\bm{\Pi}_{\mathrm{com}}^{-1}\bm{\chi}(\BI_M)]_{2} ],~\bm{\Pi}_{\mathrm{com}}=
\begin{bmatrix}
1-  \frac{L\omega\overline{\omega}}{M\delta^2} \chi(\FR,\FR)  &-\Upsilon(\FR) &-\Lambda(\FR)
\\
-  \frac{L\omega\overline{\omega}}{M\delta^2} \chi(\FR,\BF)  &1- \Upsilon(\FF) &-\Lambda(\FF)
\\
-\frac{\Xi_I}{\delta^2}  &  -\Xi\eta(\FT,\BU)   &   1-\Xi\eta(\FT,\FT) 
\end{bmatrix}.
\end{align*}
\vspace{-0.1cm}
\hrule
\end{figure*}
\begin{equation}
\begin{aligned}
\frac{R_{\mathrm{RZF}}}{K}   \xlongrightarrow[K \xlongrightarrow{(c_1,c_2)} \infty]{\mathcal{P}}  \frac{\overline{R}_{\mathrm{RZF}}}{K},
\frac{\E[R_{\mathrm{RZF}}]}{K}   \xlongrightarrow[]{K \xlongrightarrow{(c_1,c_2)} \infty}   \frac{\overline{R}_{\mathrm{RZF}}}{K},
\end{aligned}
\end{equation}
where $\overline{R}_{\mathrm{RZF}}=\sum_{k=1}^{K}   \log(1+\overline{\gamma}_{\mathrm{RZF},k})$.
\end{proposition}
\begin{proof} 
The proof of Proposition~\ref{the_ESR_common} is given in Appendix~\ref{proof_sinr_com}.
\end{proof}

 \subsubsection{ESR with ZF}
The SINR and ESR with ZF for the common correlation scenario is given in the following proposition, which can be obtained by setting $\FC_k=t_k\FC$ in Theorem~\ref{the_ZF}.

\begin{proposition} \label{the_ZF_common} (SINR and ESR with ZF) If it holds true that $ \lambda_{\mathrm{min}} (\frac{1}{M} \BH^{H}\BH) >\varepsilon $ with probability $1$ for an $\varepsilon>0$ and $\lim_{z\rightarrow 0}z\delta(z)>0$, $\lim_{z\rightarrow 0}z\omega(z)>0$ exist, the SINR can be approximated by
\begin{equation}
\label{SINR_ZF_com}
{\gamma}_{\mathrm{ZF},k} \xlongrightarrow[K \xlongrightarrow{(c_1,c_2)} \infty]{\mathcal{P}} \overline{{\gamma}}_{\mathrm{ZF},k},  
\end{equation}
where $ \overline{{\gamma}}_{\mathrm{ZF},k}=p_k (\sigma^2\sum_{l=1}^{K}\frac{p_l}{M (u_l\underline{\kappa}+t_l\underline{\omega}) })^{-1} $, and $\underline{\kappa}$ and $\underline{\omega}$ are the solution of 
 \begin{equation}
 \label{basic_zfeq2}
 \begin{aligned}
 \begin{cases}
   &\underline{\delta}=\frac{1}{M}\Tr(\bold{R}\underline{\BPS}_R),
   \\
   & \underline{\kappa}=\frac{1}{M}\Tr(\FF\underline{\BPS}_R),
   \\
 &\underline{\omega}=\frac{1}{L}\Tr( \FC\underline{\BPS}_{C}),
 \\
  &  \underline{\overline{\kappa}}=\frac{1}{L}\Tr(\BU\underline{\BPS}_{T}),
\\
 & \underline{\overline{\omega}}=\frac{1}{L}\Tr(\FT\underline{\BPS}_{T}),
\end{cases}
  \end{aligned}
 \end{equation}
  with $\underline{\BPS}_{R}=\left( \bold{I}_{M}+ \frac{L\underline{\overline{\kappa}}}{M}\BF+\frac{L \underline{\omega}\underline{\overline{\omega}}}{M\underline{\delta}}\bold{R}   \right)^{-1}$, $ \underline{\BPS}_{C}=\left(\frac{1}{\underline{\delta}}\bold{I}_{L}+\underline{\overline{\omega}}\BC\right)^{-1}$, and $\underline{\BPS}_{T} =\left(\underline{\kappa}\BU+\underline{\omega}\bold{T}\right)^{-1}$. Moreover, we have
\begin{align}
\label{eva_ZF_common}
 \frac{R_{\mathrm{ZF}}}{K}  \xlongrightarrow[K \xlongrightarrow{(c_1,c_2)} \infty]{\mathcal{P}}\frac{\overline{R}_{\mathrm{ZF}}}{K},
\frac{\E[R_{\mathrm{ZF}}]}{K}  \xlongrightarrow[]{K \xlongrightarrow{(c_1,c_2)} \infty} \frac{\overline{R}_{\mathrm{ZF}}}{K},
\end{align}
where $\overline{R}_{\mathrm{ZF}}=\sum_{k=1}^{K}  \log(1+\overline{\gamma}_{\mathrm{ZF},k})$.
\end{proposition}
\begin{proof} The proof of Proposition~\ref{the_ZF_common} is given in Appendix~\ref{proof_the_ZF_common}.
\end{proof}
\begin{remark} Propositions~\ref{the_ESR_common} and~\ref{the_ZF_common} can be obtained from Theorems~\ref{uncommon_the} and~\ref{the_ZF}, respectively, by letting $\BF_k=u_k\BF$ and $\FC_k=t_k \FC$. In this case, the $2K+1$ equations in~(\ref{basic_eq1}) degenerates to the five equations in~(\ref{basic_eq2}).
\end{remark}

Despite the above Theorems~\ref{uncommon_the} and~\ref{the_ZF}, and Propositions~\ref{the_ESR_common} and~\ref{the_ZF_common} show the evaluation of ESR for FAS-RIS systems, the impact of port selection of FAs and phase shifts of RISs on the ESR is implicit. In particular, the port selection schemes will affect correlation matrices at the BS, e.g., $\FR$ and $\BF$ in Proposition~\ref{the_ESR_common}, and the phase shifts will affect the correlation matrix at the RIS, e.g., $\BC$ in Proposition~\ref{the_ESR_common}. The ESR evaluation relies on the solution of a system equations parameterized by the correlation matrices, which makes the optimization problem challenging even if we replace $\E[R]$ with the derived ESR evaluation. Moreover, it is worth noticing that here we adopt the multi-port selection constraint in~\cite{lai2023performance}, where the potential ports of the $M$ transmit antennas can not overlap. The proposed ESR evaluation and two-timescale design can be extended to FAS-RIS systems with other multi-port selection constraints. Furthermore, the derived ESR is applicable for FAS-RIS systems with general correlation matrices and could tackle the FASs with different correlation structures.

\section{Large System Analysis Over I.I.D. Channels}
\label{large_analysis}
It has been shown that the ESR with ZF approaches that with RZF in the high signal-to-noise ratio (SNR) regime, and thus the ESR with ZF can serve as a good approximation and lower bound for the ESR with RZF. In this section, to investigate the impact of number of antennas, we ignore the correlation between FAs/reflecting elements and assume i.i.d. channels. This scenario typically occurs when antennas are deployed to be separated enough and such a methodology analyzing large-scale MIMO systems~\cite{hoydis2013massive,zhang2023asymptotic}. It is worth noticing that although the analysis with the i.i.d. setting represents an idealized case, it can be regarded as an upper bound for the practical systems. In particular. we derive the ESR with ZF over i.i.d. channels with ZF to show the impact of number of FAs and number of reflecting elements.


With i.i.d. channels, the correlation matrices are identity matrices and the large-scale channel gain for different users is the same. Thus, the i.i.d. channel can be given by $\Bh_i=\sqrt{u}\Bw_i+\sqrt{t}\BX\By_i$, where $u$ and $t$ represent the large-scale channel gain of the direct and cascaded link, respectively. An approximated expression for the ESR over i.i.d. channels is given by the following proposition.

\begin{proposition} (ESR with ZF over i.i.d. channels)
\label{iid_pro}
Under assumption~\textbf{A.1}, the ESR can be approximated by 
\begin{equation}
\label{iid_fas_esr}
\overline{R}_{\mathrm{iid,ZF}}(t,u,c_1,c_2,\sigma^2)=K\log\Bigl(1+\frac{(1-c_1)\beta(u,t,c_2)}{c_1\sigma^2}\Bigr),
\end{equation}
where $c_1=\frac{K}{M}$, $c_2=\frac{K}{L}$, and 
\begin{equation}
\begin{aligned}
\beta(u,t,c_2)&=\frac{(u+t-tc_2+\sqrt{\alpha(u,t,c_2)})}{2},
\\
\alpha(u,t,c_2)&= c_2^2t^2-2c_2t^2+t^2+2tuc_2+2tu+u^2 .
\end{aligned}
\end{equation}
\end{proposition}

\begin{proof} The proof of Proposition~\ref{iid_pro} is given in Appendix~\ref{proof_iid_pro}.
\end{proof}

Proposition~\ref{iid_pro} gives the closed-form ESR for i.i.d. RIS channels with direct link, which is equivalent to the sum of a Rayleigh channel with large-scale channel gain $u$ and a Rayleigh-product channel with large-scale channel gain $t$. Typically, we have $u>t$ due to the two-hop path loss of the cascaded link. When $t=0$, the cascaded link vanishes (no RIS) and the channel degenerates to an i.i.d. Rayleigh channel, such that~(\ref{iid_fas_esr}) degenerates to
\begin{equation}
\label{rate_ZF}
\overline{R}_{\mathrm{iid,ZF}}(0,u,c_1,c_2,\sigma^2)=K\log\left(1+\frac{(1-c_1){ u}}{c_1\sigma^2}\right),
\end{equation}
which is equivalent to~\cite[Eq. (44)]{wagner2012large}. Next, we show the gain due to the utilization of the RIS.
\begin{remark} (\textbf{How many selected ports do we need for FAS-RIS systems?})
We can also use Proposition~\ref{iid_pro} to evaluate the number of selected ports required to achieve a given target rate $R_{\mathrm{target}}$ with given $L$ and $K$. By~(\ref{iid_fas_esr}), we can obtain that the minimum number of selected ports $M^{*}$ as
\begin{equation}
M^{*}=K\Bigl(\frac{\sigma^2}{\beta(u,t,c_2)}( 2^{\frac{R_{\mathrm{target}}}{K}}-1)+1\Bigr).
\end{equation}
Note that $M^{*}$ is derived based on i.i.d. channel assumption and assumption~\textbf{A.1}, where the channel correlation is not taken into account. Although the i.i.d. may lead to in an optimistically small required number of ports, it provides a lower bound for the number of selected ports $M$. Compare with FASs without RIS, FAS-RIS systems require less selected ports.
\end{remark}
\begin{remark}(\textbf{Impact of number of reflecting elements at the RIS})
Under assumption~\textbf{A.1}, the SINR gain due to the cascaded link is $\frac{(1-c_1)(t-tc_2-u+\sqrt{\alpha(u,t,c_2)})}{2c_1\sigma^2}$,
which increases with the number of reflecting elements. When $L \gg K$ ($c_2\rightarrow 0$), i.e., the number of elements at the RIS is much larger than the number of users and selected ports, the gain approaches $\frac{(1-c_1)t}{c_1\sigma^2}$ such that the limiting sum rate is  
\begin{equation}
\label{rate_ZF}
\overline{R}_{\mathrm{iid,ZF}}(t,u,c_1,0,\sigma^2)=K\log\left(1+\frac{(1-c_1) (u+t)}{c_1\sigma^2}\right),
\end{equation}
which can be regarded as the ESR over an i.i.d. Rayleigh channel with gain $u+t$. This is because as the number of reflecting elements goes to infinity, the statistical attribute of the cascaded link (Rayleigh-product channel) approaches that of a single-hop Rayleigh channel~\cite{zhang2023asymptotic}.
\end{remark}




\begin{remark} 
\label{rem_comp_mrt}
(\textbf{Comparison with MRT precoding}) Given assumption~\textbf{A.1}, the ESR with MRT precoding can be approximated by
\begin{align*}
&\overline{R}_{\mathrm{MRT}} \approx  \numberthis \label{iid_eva_MRT}
\\
& K\log\Bigl( 1+\frac{ ( t +u  )^2}{\frac{(K-1)t(u+t)}{M}+\frac{(K-1)t(\frac{t L }{M^2} + \frac{uL}{M} )}{L}+\frac{K\sigma^2(t+u)}{M}} \Bigr).
\end{align*}
We can conclude from~(\ref{iid_eva_MRT}) that when the SNR increases, the ESR with MRT precoding will saturate and not increase with the SNR while the ESR of ZF increases with the SNR. This shows that ZF and RZF outperforms MRT in the high SNR regime.
\end{remark}

\section{ESR Optimization}
\label{sec_algorithm}
In this section, we tackle the ESR maximization problem $\mathcal{P}1$ in~(\ref{p1_exp}) by a two-timescale approach. In particular, the port selection, regularization factor $z$, and phase shifts $\bold{\Phi}$ are optimized based on the statistical CSI, i.e., correlation matrices at the BS and RIS and large-scale channel gain, while the instantaneous CSI for the end-to-end link is utilized for RZF/ZF precoding. In each iteration, we first determine the port selection vector $\Bs$ with given $z$ and $\bm{\Phi}$ as shown in Section~\ref{sec_port_sel}. Then we jointly optimize $z$ and $\bm{\Phi}$ with given $\Bs$ in Section~\ref{sec_RZF_opt} and give the overall algorithm in Section~\ref{sec_overall_alg}. Section~\ref{sec_homo_z} gives the optimal $z$ for homogeneous FAS-RIS systems.
\subsection{Port Selection}
\label{sec_port_sel}
Assuming $M\ge K$, the port selection problem with given phase shifts and regularization factor can be formulated as the following optimization problem
\begin{equation}
\begin{aligned}
\mathcal{P}2:~& \max_{\bm{s}} \overline{R}_{\mathrm{RZF}}(\Bs),
\\
\mathrm{s.t.}~& \mathcal{C}_1:\Bs^{T}\bold{1}_{M_{\mathrm{tot}}}=M,~\Bs\in\{0,1\}^{M_{\mathrm{tot}}}.\\
\end{aligned}
\end{equation}
Note that $\mathcal{P}2$ is non-convex and challenging to solve due to the port selection constraint, whose optimal solution requires the exhaustive search. Even if we adopt the linear relaxation~\cite{dua2006receive} to replace the binary constraint with linear constraint on continuous $\Bs$, the problem is still non-convex due to the non-convexity of the objective function $\overline{R}_{\mathrm{RZF}}(\Bs)$. To this end, we resort to optimize a lower bound $\overline{R}_{\mathrm{ZF}}(\Bs)$. Moreover, it can be proved that  $\overline{R}_{\mathrm{ZF}}(\Bs)$ is convex with respect to $\Bs \in [0,1]^{M_{\mathrm{tot}}}$. This substitution facilitates linear relaxation of the port selection vector, allowing the FW algorithm to be applicable for the port selection. Consequently, this approach motivates us to find a suitable port selection scheme based on ZF precoding and further optimize the phase shifts and regularization factor $\overline{R}_{\mathrm{RZF}}(\Bs)$ with given $\Bs$. The port selection problem with ZF can be formulated as follows.

\begin{equation}
\begin{aligned}
\mathcal{P}3:~& \max_{\bm{s}} \overline{R}_{\mathrm{ZF}}(\Bs),
\\
\mathrm{s.t.}~& \overline{\mathcal{C}}_1:\Bs^{T}\bold{1}_{M_{\mathrm{tot}}}\le M,~\Bs \in [0,1]^{M_{\mathrm{tot}}}.\\
\end{aligned}
\end{equation}
The asymptotic convexity of $\overline{R}_{\mathrm{ZF}}(\Bs)$ follows from
\begin{align*}
&\frac{1}{M}\Tr((\BP^{\frac{1}{2}}\BH^{H}_{\mathrm{tot}}\diag(\Bs)\BH_{\mathrm{tot}}\BP^{\frac{1}{2}})^{-1})
\\
&\xlongrightarrow[K \xlongrightarrow{(c_1,c_2)} \infty]{\mathcal{P}} \sum_{l=1}^{K}\frac{p_l}{M (u_l\underline{\kappa}+t_l\underline{\omega}) },\numberthis\label{obj_convex}
\end{align*}
with $\BH_{\mathrm{tot}}=\LF_{\mathrm{tot}}\BW'\RU+\LR_{\mathrm{tot}}\BX'\MC_{L} \bold{\Phi}\MC_{R}\BY\RT$, and $\BW'\in\mathbb{C}^{M_{\mathrm{tot}}\times K}$ and $\BX'\in \mathbb{C}^{M_{\mathrm{tot}}\times L}$ are i.i.d. Gaussian random matrices. The left hand side of~(\ref{obj_convex}) is convex with respect to $\Bs$, which indicates the asymptotic convexity of $\sum_{l=1}^{K}\frac{p_k}{M (u_l\underline{\kappa}+t_l\underline{\omega}) }$ and thus $\overline{R}_{\mathrm{ZF}}(\Bs)$ is convex. The proof of the asymptotic convexity is similar to~\cite[Theorem 1]{elkhalil2018measurement} and is omitted here. To solve $\mathcal{P}3$, we use the FW method~\cite{bertsekas1997nonlinear} and exploit the special structure of the constraint on $\Bs$ to obtain the closed-form update rule for $\Bs$.  Specifically, given the current iteration $\Bs^{(t)}$, in the next iteration of the FW method, it requires to solve the following problem to determine the update rule
\begin{equation}
\label{update_step_pro}
\overline{\Bs}^{(t+1)}=\arg \max_{\Bs}~(\Bs-{\Bs}^{(t)})^{T}\bm{\nabla}_{\Bs}\overline{R}_{\mathrm{ZF}}(\Bs),~\mathrm{s.t.}~\overline{\mathcal{C}}_1,
\end{equation}
where the gradient $\bm{\nabla}_{\Bs}\overline{R}_{\mathrm{ZF}}(\Bs)$ can be computed by the same approach as shown in Section~\ref{sec_RZF_opt}. With the special structure of the problem in~(\ref{update_step_pro}), we can derive a lightweight algorithm with closed-form optimal solution to the subproblem in~(\ref{update_step_pro}), which can be obtained by setting $s_i=1$ if $[\bm{\nabla}_{\Bs}\overline{R}_{\mathrm{ZF}}(\Bs)]_i$ is among the largest $M$ values of $[\bm{\nabla}_{\Bs}\overline{R}_{\mathrm{ZF}}(\Bs)]_l$, $l=1,2,...,M_{\mathrm{tot}}$, and $s_i=0$ otherwise. The port selection algorithm is presented in \textbf{Algorithm~\ref{PORT_alg}}.
\begin{algorithm}[t!]
\caption{Port Section Algorithm}
\label{PORT_alg} 
\begin{algorithmic}[1] 
\REQUIRE  Set $\Bs^{(0)}=\frac{M}{M_{\mathrm{tot}}}\bm{1}_{M_{\mathrm{tot}}}$ and $t=0$.
\REPEAT
\STATE Compute $\overline{\bm{s}}^{(t+1)}$ by solving~(\ref{update_step_pro}).
\STATE ${\bm{s}}^{(t+1)}={\bm{s}}^{(t)}+\frac{2}{t+2}(\overline{\bm{s}}^{(t+1)}-{\bm{s}}^{(t)})$.
\STATE $t \leftarrow  t+1$.
\UNTIL $\Bigl| \frac{\overline{R}_{\mathrm{ZF}}(\Bs^{(t)})-\overline{R}_{\mathrm{ZF}}(\Bs^{(t-1)})}{\overline{R}_{\mathrm{ZF}}(\Bs^{(t-1)})} \Bigr|<\varepsilon$.
\STATE Find the index set $\mathcal{M}$ of $M$ largest values in $\Bs^{(t)}$. 
\STATE Let $\Bs^{*}$ be $[\Bs^{*}]_{m}=1$ if $m\in \mathcal{M}$ and $[\Bs^{*}]_{m}=0$ if $m\notin \mathcal{M}$.
\ENSURE  $\Bs^{*}$.
\end{algorithmic}
\end{algorithm}

\subsection{RZF and Phase Shifts Optimization With Given Port Selection}
\label{sec_RZF_opt}
Next, we investigate the design for $z$ and $\bold{\Phi}$ with fixed $\Bs$ and and exploit the alternating optimization (AO) approach to solve $\mathcal{P}1$. The optimization problems for $z$ and $\bold{\Phi}$ can be given by following $\mathcal{P}4$ and $\mathcal{P}5$, respectively,
\begin{align}
\mathcal{P}4:~& \max_{z>0} \overline{R}_{\mathrm{RZF}}(z),\label{p2_exp}
\\
\mathcal{P}5:~& \max_{\bold{\Phi}} \overline{R}_{\mathrm{RZF}}(\bm{\Phi}),~~\nonumber
\\
&
\mathrm{s.t.}~\mathcal{C}2:~\bold{\Phi}=\diag(e^{\jmath \phi_1},e^{\jmath \phi_2},...,e^{\jmath \phi_L}).\label{p3_exp}
\end{align}
In each iteration of the AO algorithm, we iteratively solve $\mathcal{P}4$ and $\mathcal{P}5$, by employing the $1$-dimensional search and gradient ascent algorithm, respectively. The gradient algorithm for solving $\mathcal{P}5$ and AO algorithm is given in~\textbf{Algorithm~\ref{gra_alg}} and~\textbf{Algorithm~\ref{AO_alg}}, respectively. In each iteration of~\textbf{Algorithm~\ref{gra_alg}}, we update $\phi_l$ by searching in the gradient direction. The search is terminated when a stationary point for the objective function is obtained. If ZF is adopted to replace RZF, we can omit the optimization of $z$ and only need to optimize the phase shifts by gradient methods.

\begin{algorithm}[t!]
\caption{Gradient Ascent Algorithm by Optimizing $\bold{\Phi}$}
\label{gra_alg} 
\begin{algorithmic}[1] 
\REQUIRE  $\bm{\Phi}^{\left(0 \right)}$, initial stepsize $\alpha_{0}$, scaling factor $0<c<1$, and control parameter $0<\beta<1$. 
Set $t=0$.
\REPEAT
\STATE Compute the gradient 
$$\bm{\nabla}_{\bm{\phi}}\overline{R}(\boldsymbol{\Phi})=  (\frac{\partial \overline{R}(\boldsymbol{\Phi})}{\partial \phi_{1}}, \frac{\partial \overline{R}(\boldsymbol{\Phi})}{\partial \phi_{2}},..., 
\frac{\partial \overline{R}(\boldsymbol{\Phi})}{\partial \phi_{L}})^{T},$$ 
using~(\ref{derivative_Rate}) and its direction $\bold{g}^{(t)}=\frac{\bm{\nabla}_{\boldsymbol{\Phi}}\overline{R}(\boldsymbol{\Phi})}{\|\bm{\nabla}_{\boldsymbol{\Phi}}\overline{R}(\boldsymbol{\Phi}) \|}$.
\STATE $\alpha\leftarrow\alpha_{0}$.
\WHILE{$\overline{R}(\diag[\exp(\jmath\bm{\phi}^{(t)}+\alpha \jmath \bold{g}^{(t)})])-\overline{R}(\boldsymbol{\Phi}^{(t)})<\alpha\beta \| \bm{\nabla}_{\bm{\Phi}}\overline{R}(\bm{\Phi}^{(t)})\|$} 
\STATE  $\alpha \leftarrow c\alpha$.
\ENDWHILE 
	\STATE $\boldsymbol{\phi}^{(t+1)} \leftarrow \bm{\bold{\phi}}^{(t)}-\alpha \bold{g}^{(t)}$.
	\STATE $\bold{\Phi}^{(t+1)} \leftarrow \diag[\exp(\jmath\bm{\phi}^{(t+1)})]$.
\STATE $t \leftarrow  t+1$.
\UNTIL $\Bigl|\frac{\overline{R}^{(t)}(\bold{\Phi})-\overline{R}^{(t-1)}(\bold{\Phi})}{\overline{R}^{(t-1)}(\bold{\Phi})}\Bigr|<\varepsilon$.
\ENSURE  $ \bold{\Phi}^{(t)}$.
\end{algorithmic}
\end{algorithm}

\begin{algorithm}[t!]
\caption{RZF and Phase Shifts Optimization with Given Selection Vector}
\label{AO_alg} 
\begin{algorithmic}[1] 
\REQUIRE  $\bm{\Phi}^{\left(0 \right)}$, $z^{(0)}$. 
Set $t=0$.
\REPEAT
\STATE Fix $\boldsymbol{\phi}^{(t-1)}$ and solve $\mathcal{P}4$ to obtain $z^{(t)}$ by $1$-dimensional search. 
\STATE Fix $z^{(t)}$ and solve $\mathcal{P}5$ by Algorithm~\ref{gra_alg}.
\STATE $t \leftarrow  t+1$.
\UNTIL Convergence.
\ENSURE  $\bm{\Phi}^{(t)}, z^{(t)}$.
\end{algorithmic}
\end{algorithm}

Since the phase shifts optimization requires the closed-form expression for the partial derivatives, we compute the partial derivatives of the ESR with respect to phase shifts $\phi_{l}$, $l=1,2,...,L$. In the following, for simplicity, we adopt the notation $(\cdot)^{(l)}=\frac{\partial (\cdot)}{\partial \phi_{l}}$. The approximation for the ESR is computed based on the parameters $\delta$, $\omega$, and $\overline{\omega}$ and the partial derivative of $R^{(l)}(\bm{\Phi})$ can be computed by the chain rule. We first evaluate $\delta^{(l)},\omega^{(l)},\overline{\omega}^{(l)}$ and then compute $R^{(l)}_{\mathrm{RZF}}(\bm{\Phi})$.

\subsubsection{Computation of $\delta^{(l)},{\kappa}^{(l)},\omega^{(l)}$}
Note that according to~(\ref{basic_eq2}), the parameters $\delta$, $\kappa$, $\omega$, $\overline{\kappa}$, $\overline{\omega}$ are implicit functions of $\phi_l$. The derivatives $\delta^{(l)}$, $\kappa^{(l)}$, $\omega^{(l)}$ can be computed by solving a system of equations. In fact, since $(\delta,\kappa, \omega)$ is the positive solution of~(\ref{basic_eq2}), we can compute $\bold{v}_{l}=(\delta^{(l)}, {\kappa}^{(l)},\omega^{(l)})^{T}$ by taking the derivative with respect to $\phi_l$ on both sides of~(\ref{basic_eq2}) to obtain
\begin{equation}
\begin{aligned}
\bm{\Pi}_{\mathrm{com}}\bold{v}_{l}
=\bold{u}_{l},
\end{aligned}
\end{equation}
 with $
\bold{u}_{l}
=
\begin{bmatrix}
0 &0 & U(\BA_l)
\end{bmatrix}^{T}
$, $U(\BA_l)=\frac{\Tr(\BA_l\BPS_C)}{L}-\frac{\overline{\omega}}{L}\Tr(\BA_l\BPS_C\FC\BPS_C)$, $\BA_l=\LC_L(\bold{G}_{l}\otimes\FC_R)\RC_L$, and 
\begin{equation}
\label{gg_pq}
\left[\BG_l\right]_{p,q}=\left\{
\begin{aligned}
& \jmath e^{\jmath (\phi_{l}-\phi_{q})} ,&p = l, \\
& -\jmath e^{\jmath (\phi_{p}-\phi_{l})} ,&q = l, \\
&0,  &\mathrm{otherwise}.
\end{aligned}
\right.
\end{equation}
Therefore, we have $\bold{u}_{l}=\bm{\Pi}_{\mathrm{com}}^{-1}\bold{v}_{l},  ~ l=1,2,...,L.$

\subsubsection{Computation of Gradients} The derivative $\overline{R}^{(l)}_{\mathrm{RZF}}(\boldsymbol{\Phi})$ can be computed based on $\bold{u}_{l}$ by the chain rule and the expression is given by
\begin{equation}
\label{derivative_Rate}
\overline{R}^{(l)}_{\mathrm{RZF}}(\boldsymbol{\Phi})=\sum_{k=1}^{K}  \frac{\overline{\gamma}_k^{(l)}}{1+\overline{\gamma}_k},
\end{equation}
where $\overline{\gamma}_k^{(l)}$ can be computed according to the chain rule and the closed-form expression is given in~(\ref{gamma_der}) at the top of the next page, 
\begin{figure*}[t!]
\vspace{-0.5cm}
\begin{align*}
\BPS_R^{(l)}&=-\BPS_R\Bigl(z\BI_M+\frac{L\FR}{M}[\frac{\omega^{(l)}\overline{\omega}}{\delta} +\frac{\omega}{\delta} (\omega^{(l)}\eta(\FT,\FT)+\kappa^{(l)}\eta(\FT,\BU))-\frac{\omega^{(l)}\overline{\omega}\delta^{(l)}}{\delta^2}]+ \frac{L\FF}{M}( \omega^{(l)}\eta(\FT,\BU)+\kappa^{(l)}\eta(\BU,\BU)) \Bigr) \BPS_R,
\\
\BPS_C^{(l)}&=-\BPS_C\Bigl( -\frac{\delta^{(l)}}{\delta^2}\BI_L +(\omega^{(l)}\eta(\FT,\FT)+\kappa^{(l)}\eta(\BU,\FT))\FC+\overline{\omega}_l\BA_l \Bigr) \BPS_C,
~\BPS_T^{(l)}=-\BPS_T( \BI_K +\omega^{(l)}\FT+\kappa^{(l)}\BU ) \BPS_T,
\\
\chi^{(l)}(\BA,\BB)&=\frac{1}{M}(\Tr(\BA\BPS_R^{(l)}\BB\BPS_R)+\Tr(\BA\BPS_R\BB\BPS_R^{(l)}) ),~\eta^{(l)}(\BA,\BB)=\frac{1}{L}(\Tr(\BA\BPS_T^{(l)}\BB\BPS_{T})+\Tr(\BA\BPS_T\BB\BPS_{T}^{(l)})),
\\
\Xi^{(l)}&=\frac{2}{L}(\Tr(\BA_l\BPS_C\BC\BPS_C)-\Tr(\BC\BPS_C\BC\BPS_C\BPS_C^{(l)} \BPS_C)),~\Xi_I^{(l)}=\frac{1}{L}(\Tr(\BA_l\BPS_C^2)-2\Tr(\BC\BPS_C^2\BPS_C^{(l)} \BPS_C)),
\\
\Delta^{(l)} &=-\Xi^{(l)}\eta(\FT,\FT)-\Xi\eta^{(l)}(\FT,\FT),~[\bm{\Pi}_{\mathrm{com}}^{-1}\bm{\chi}(\BA)]^{(l)}=-\bm{\Pi}_{\mathrm{com}}^{-1}\bm{\Pi}_{\mathrm{com}}^{(l)} \bm{\Pi}_{\mathrm{com}}^{-1}\bm{\chi}(\BA)+\bm{\Pi}_{\mathrm{com}}^{-1}\bm{\chi}^{(l)}(\BA),
\\
\Psi_{k,m}^{(l)}&=-\frac{t_m t_k \Delta^{(l)}}{\Delta^2}(\Xi+\frac{L\Xi_I }{M\delta^2} [\bm{\Pi}_{\mathrm{com}}^{-1}\bm{\chi}(\FR)]_{3})
+\frac{t_m t_k}{\Delta}(\Xi^{(l)}+\frac{L\Xi_I }{M\delta^2} [\bm{\Pi}_{\mathrm{com}}^{-1}\bm{\chi}(\FR)]_{3}^{(l)}+ [\bm{\Pi}_{\mathrm{com}}^{-1}\bm{\chi}(\FR)]_{3}( \frac{L\Xi_I^{(l)} }{M\delta^2})-\frac{2L\Xi_I\delta^{(l)} }{M\delta^3})
\\
&+\frac{Lt_m t_k}{M\Delta}( \Xi^{(l)}\eta(\FT,\BU)+\Xi\eta^{(l)}(\FT,\BU)-\Xi\eta(\FT,\BU)\frac{\Delta^{(l)}}{\Delta} )[\bm{\Pi}_{\mathrm{com}}^{-1}\bm{\chi}(\BF)]_{3}
\\
&
+(\frac{ Lt_l t_k\Xi\eta(\FT,\BU)}{M\Delta}+\frac{L (t_m u_k+t_k u_m)}{M}) [\bm{\Pi}_{\mathrm{com}}^{-1}\bm{\chi}(\BF)]_{3}^{(l)}+\frac{u_m u_k L}{M} [\bm{\Pi}_{\mathrm{com}}^{-1}\bm{\chi}(\BF)]_{2}^{(l)} ,
\\
\overline{C}_{\mathrm{com}}&=\frac{L}{M}[\eta^{(l)}(\BP,\FT)[\bm{\Pi}_{\mathrm{com}}^{-1}\bm{\chi}(\BI_M)]_{3}+ \eta(\BP,\FT)[\bm{\Pi}_{\mathrm{com}}^{-1}\bm{\chi}(\BI_M)]^{(l)}_{3}+\eta^{(l)}(\BP,\BU)[\bm{\Pi}_{\mathrm{com}}^{-1}\bm{\chi}(\BI_M)]_{2}+\eta^{(l)}(\BP,\BU)[\bm{\Pi}_{\mathrm{com}}^{-1}\bm{\chi}(\BI_M)]_{2}^{(l)}  ],
\\
\overline{\gamma}_k^{(l)} &=\frac{2\overline{\gamma}_k(t_k\omega^{(l)}+ u_k\kappa^{(l)})}{t_k\omega+u_k\kappa}
- \frac{\overline{\gamma}_k^2}{p_k(t_k\omega+u_k\kappa)}\Bigl[ \sum_{k\neq m }  p_m \Bigl(\frac{\Psi_{k,m}^{(l)}}{L(1+t_m\omega+u_m\kappa)^2}-\frac{2\Psi_{k,m}(t_m\omega^{(l)}+u_m\kappa^{(l)})}{L(1+t_m\omega+u_m\kappa)^3}\Bigr) 
\\
&
+2(1+t_k\omega+u_k\kappa)(t_k\omega^{(l)}+u_k\kappa^{(l)})\overline{C}_{\mathrm{com}}+(1+t_k\omega+u_k\kappa)^2\overline{C}_{\mathrm{com}}^{(l)} \Bigr].
\numberthis \label{gamma_der}
\end{align*}
\hrulefill
\vspace{-0.4cm}
\end{figure*}
where $\chi(\BA,\BB)=\frac{1}{M}\Tr(\BA\BPS_R\BB\BPS_{R})$ and $\eta(\BA,\BB)=\frac{1}{L}\Tr(\BA\BPS_{T}\BB\BPS_{T})$ with arbitrary $\BA$ and $\BB$.

\subsubsection{Extension to Uncommon Correlation Case}
We can generalize the phase shift optimization method of the common correlation case to the uncommon correlation case by computing the derivative for the SINR of uncommon correlation case. The derivative of the key parameters with respect to $\phi_l$, i.e., $\mu_k^{(l)}=\frac{\partial \mu_k }{\partial \phi_{l}}$ and $\delta^{(l)}=\frac{\partial \delta }{\partial \phi_{l}}$, can be obtained by taking derivative on both sides of~(\ref{basic_eq1}) and explicitly expressed by the solution $\bold{w}_{l}=\begin{bmatrix}\mu_1^{(l)},\mu_2^{(l)},...,\mu_K^{(l)}, \delta^{(l)}   \end{bmatrix}^{T}$ for the following system of equations
\begin{equation}
\begin{aligned}
\bm{\Pi}\bold{w}_{l}
=\bold{n}_{l},
\end{aligned}
\end{equation}
 with $
\bold{n}_{l}
=
\begin{bmatrix}
E_1(\BA_1^{(l)}),E_2(\BA_2^{(l)}),...,E_K(\BA_K^{(l)}),0
\end{bmatrix}^{T}
$, $E_k(\BA_l^{(l)})=\frac{\Tr(\BA_l^{(l)}\BPS_C)}{L}-\sum_{k=1}^{K}\frac{1}{L^2(1+\mu_k)}\Tr(\BA_k^{(l)}\BPS_C\FC\BPS_C)$, $\BA_k^{(l)}=\LC_{L,k}(\bold{G}_{l}\otimes\FC_{R,k})\RC_{L,k}$, and $\BG_l$ is given in~(\ref{gg_pq}). The closed-form evaluation for the derivative of SINR can be obtained by the chain rule by following similar computations of common correlation case.

\subsection{Joint Optimization Algorithm}
\label{sec_overall_alg}
The overall algorithm is given in \textbf{Algorithm~\ref{overall_alg}}.
\begin{algorithm}[t!]
\caption{Joint Optimization of Port Selection, Regularization Factor, and Phase Shifts}
\label{overall_alg} 
\begin{algorithmic}[1] 
\REQUIRE  $\FR_{\mathrm{tot}}$, $\BF_{\mathrm{tot}}$, $\BC_L$, $\BC_R$, $\FT$, $\BU$, $\bm{\Phi}^{(0)}$, $z^{(0)}$, $T_{\mathrm{iter}}$, $\mathcal{R}=\{ \}$,  and $t=1$.
\REPEAT
\STATE Compute $\Bs^{(t)}$ using Algorithm~\ref{PORT_alg} with $\bm{\Phi}^{(t-1)}$, $z^{(t-1)}$.
\STATE Compute $\bm{\Phi}^{(t)}$, $z^{(t)}$, and $R^{(t)}=\overline{R}_{\mathrm{RZF}}(\Bs^{(t)},\bm{\Phi}^{(t)},z^{(t)})$ using Algorithm~\ref{AO_alg} with $\Bs^{(t)}$.
\STATE $\mathcal{R}=\mathcal{R}\cup \{R^{(t)} \}$
\STATE $t \leftarrow  t+1$.
\UNTIL  $t >T_{\mathrm{iter}}$.
\STATE Find the index $i$ for the largest element in $\mathcal{R}$.
\ENSURE  $\Bs^{(i)}$, $z^{(i)}$, and $\bm{\Phi}^{(i)}$.
\end{algorithmic}
\end{algorithm}


\subsection{Do We Need to Optimize Regularization Factor for Homogeneous FAS-RIS Systems?}
\label{sec_homo_z}
In general cases, we resort to the search approach to find the optimal $z$. However, it can be proved that when users share common statistical CSI, i.e., the channel is homogeneous, the optimal $z$ is given by the following proposition.

\begin{proposition}\label{pro_opt_z} When the statistical CSI and signal power of users are same, i.e., common $\BF$, $\FS$, and $\FT=t\BI_K$, $\BU=u\BI_K$, $\BP=\BI_K$, the optimal $z$ that maximizes the ESR is $\frac{K\sigma^2}{M}$.
\end{proposition}

\begin{proof} The proof of Proposition~\ref{pro_opt_z} is given in Appendix~\ref{proof_pro_opt_z}.
\end{proof}

Proposition~\ref{pro_opt_z} indicates that for a FAS-RIS system with homogeneous channels, the optimal $z$ that maximizes the ESR does not depend on the correlation matrices and channel gains, and can be given directly without optimization. This decouples the AO of $z$ and $\bm{\Phi}$.

\section{Simulations}
In this section, we will validate the accuracy of ESR evaluation in Section~\ref{sec_res} and illustrate the performance of the proposed port selection algorithm in Section~\ref{sec_algorithm} by numerical simulations. In the figures, the theoretical values and Monte Carlo simulations are represented by curves and markers, respectively.
\label{sec_simu}
\subsection{Simulation Settings}

\textit{Correlation matrix of RIS:} Here we consider a RIS consisting of a linear array with uniformly distributed angle spreads and the correlation matrix of RIS is generated is given by $\BC\in \mathbb{C}^{L\times L}$~\cite{kammoun2019asymptotic},
$ [\BC(d_c, \alpha, \beta, L)]_{m,n}
 = \int_{-180}^{180} \frac{1}{\sqrt{2\pi\beta^2}} 
 e^{\jmath \frac{2\pi}{\lambda} d_c(m-n)\sin(\frac{\pi\phi}{180})-\frac{(\phi-\alpha)^2}{2\beta^2} } \mathrm{d}\phi $,~m,n =1,2,...,L,
where $d_c=0.5$ represents the antenna spacing measured in wavelength, and $\alpha$, and $\beta^2$ denote the mean angle and the mean-square angle spreads, measured in degrees, respectively.
  
\textit{Path loss model:} 
The path loss of BS-RIS link, the RIS-user link, and direct link (BS-user) are given by~\cite{danufane2021path}
$P_{\mathrm{BS-RIS}}=\frac{C_{\mathrm{BS-user}}}{d_{\mathrm{BS-RIS}}^{\alpha_{\mathrm{BS-RIS}}}},~P_{\mathrm{RIS-user}}=\frac{C_{\mathrm{BS-user}}}{d_{\mathrm{RIS-user}}^{\alpha_{\mathrm{RIS-user}}}}, 
P_{\mathrm{BS-user}}=\frac{C_{\mathrm{BS-user}}}{d_{\mathrm{BS-user}}^{\alpha_{\mathrm{BS-user}}}}$,
where $d_{\mathrm{BS-user}}$, $d_{\mathrm{BS-RIS}}$, and $d_{\mathrm{RIS-user}}$ represent the distances, and $\alpha_{\mathrm{BS-RIS}}=\alpha_{\mathrm{RIS-user}}=2.1$, and $\alpha_{\mathrm{BS-user}}=3.2$ denote the path loss exponents of associated links. Here $C_{\mathrm{BS-user}}=C_{\mathrm{RIS-user}}=C_{\mathrm{BS-RIS}}=-20$ dB represent the reference path-loss at $1$ meter.

\subsection{Accuracy of ESR Evaluation}

\textit{ESR with RZF of Uncommon Correlation Case}: Fig.~\ref{uncommon_eva} depicts the ESR with $M=\{16,20, 24\}$, $L=32$, $K=12$, $z=\frac{K\sigma^2}{M}$, $d_{\mathrm{BS-RIS}}=5$ m, and $d_{\mathrm{RIS-}i}= 20+\lfloor\frac{i-1}{2} \rfloor$ m. The angle between BS-RIS and RIS-user links is $120^{\circ}$ such that the distance between the BS and the $k$-th user $d_{\mathrm{BS-RIS}}$ can be computed by the cosine formula $d_{\mathrm{BS-k}}=\sqrt{ d_{\mathrm{BS-RIS}}^2+d_{\mathrm{RIS-}k}^2-2\cos(150^{\circ})d_{\mathrm{BS-RIS}}d_{\mathrm{RIS-}k}  }$.
The correlation matrices are set as $\BF_i=\BC(0.5,10+2(i-1),30,M)$, $\FC_{R,i}=\BC(0.5,5+10(i-1),30,32)$, $i=1,2,...,K$, $\FC_{L}=\BC(0.5,5,30,32)$, and $\FR=\BC(0.5,10,5,M)$. Simulation values are obtained by $10000$ Monte-Carlo realizations. It can be observed that the proposed ESR evaluation matches the simulation values well for both RZF and ZF precoding, which validates the accuracy of~(\ref{eva_uncommon}) in Theorem~\ref{uncommon_the}. With the same setting and $\sigma^{-2}=80$ dB, Fig.~\ref{fig_uncommon_size} presents the ESR with different system sizes, where $(M, K,L)$ proportionally scales up from $(8,6,16)$ and $(12,6,16)$ in Cases 1 and 2, respectively. It can be observed that the proposed ESR evaluation is accurate even when the systems size is small.


\begin{figure*}
\begin{minipage}[t]{0.33\textwidth}
\centering\includegraphics[width=1\textwidth]{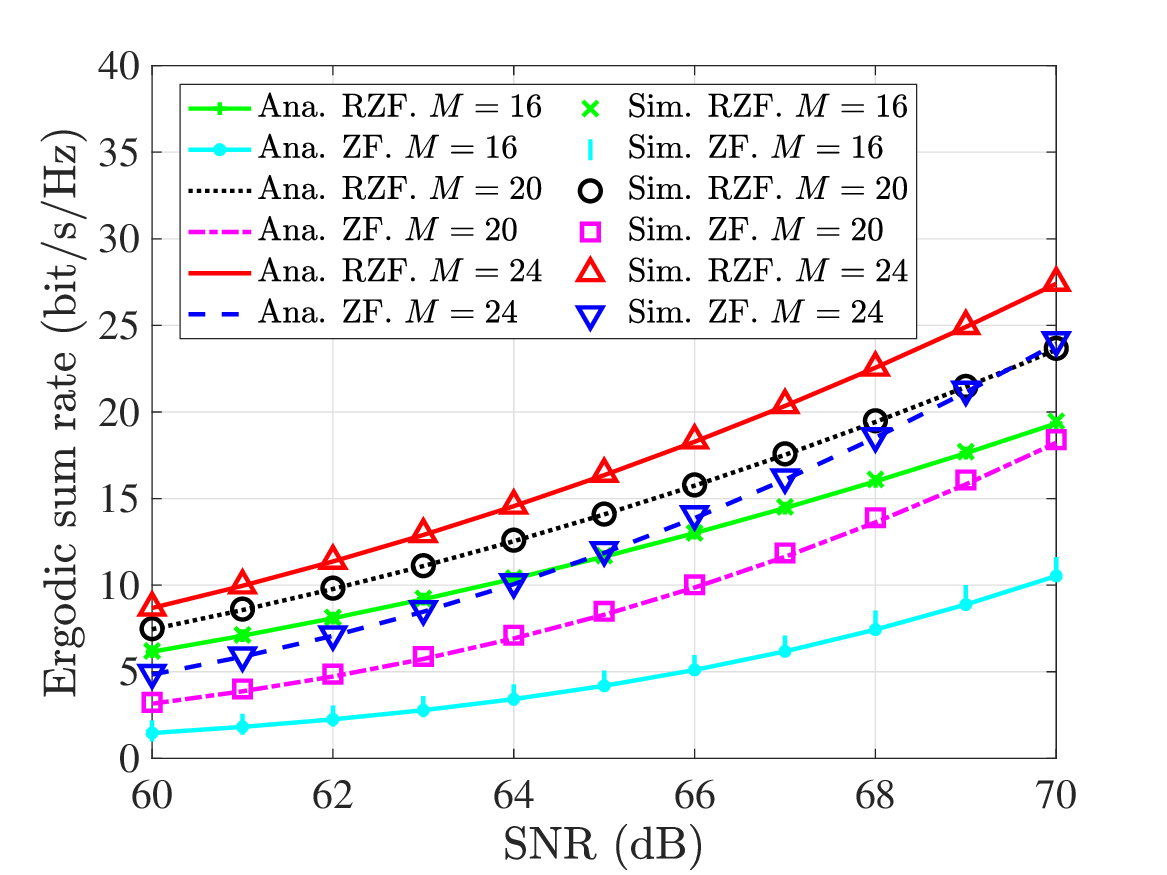}
\captionof{figure}{ESR with uncommon scenario.}
\label{uncommon_eva}
\vspace{-0.4cm}
\end{minipage}
\begin{minipage}[t]{0.33\textwidth}
\centering\includegraphics[width=1\textwidth]{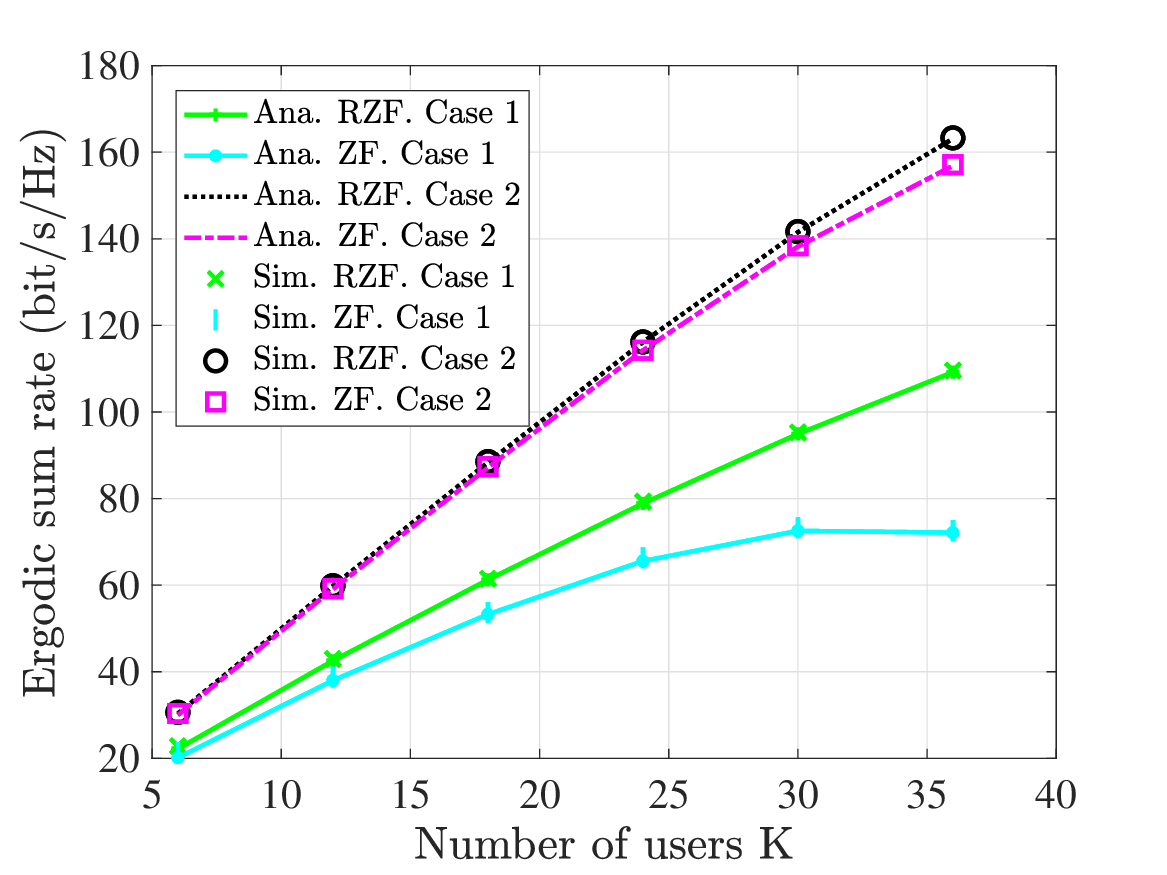}
\captionof{figure}{ESR vs system size.}
\label{fig_uncommon_size}
\vspace{-0.4cm}
\end{minipage}
\begin{minipage}[t]{0.33\textwidth}
\centering\includegraphics[width=1\textwidth]{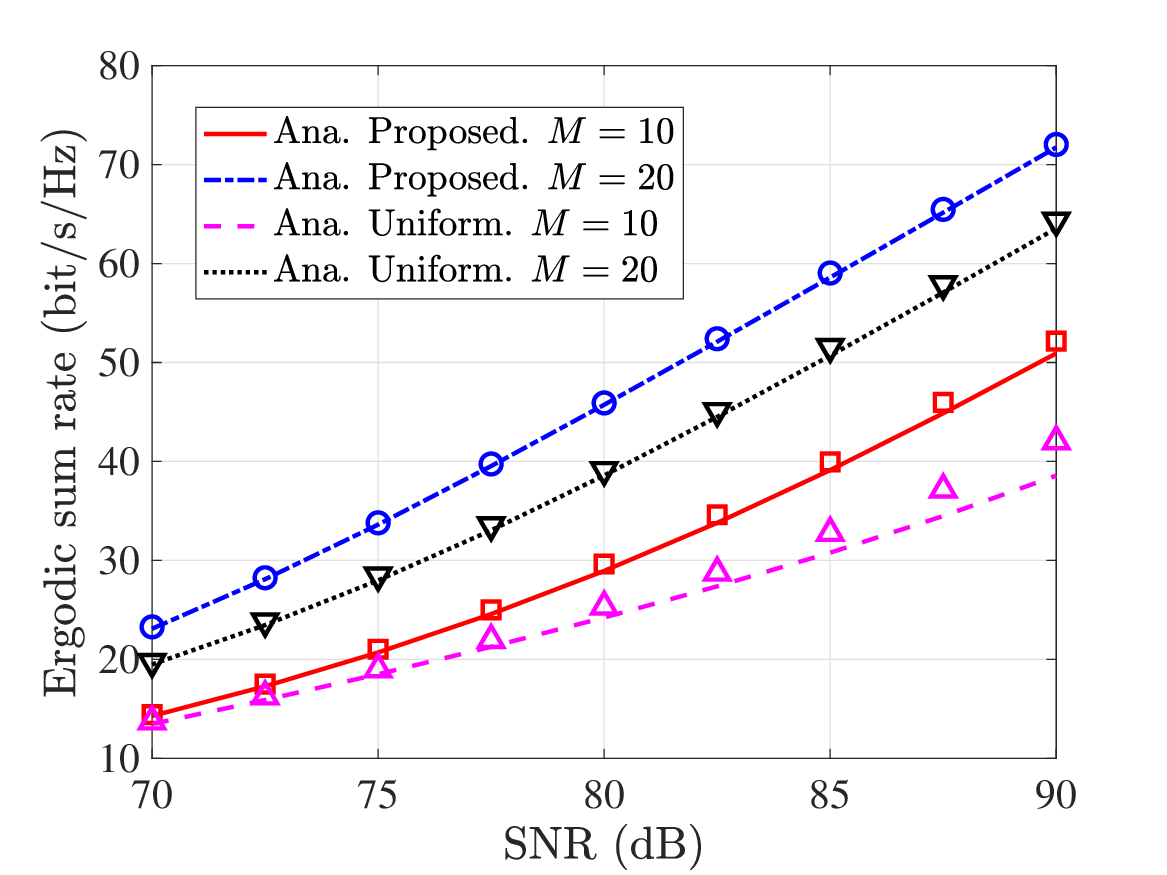}
\captionof{figure}{Optimization performance of proposed scheme.}
\label{fig_opt}
\vspace{-0.4cm}
\end{minipage}

\end{figure*}

\begin{figure*}
\begin{minipage}[t]{0.33\textwidth}
\centering\includegraphics[width=1\textwidth]{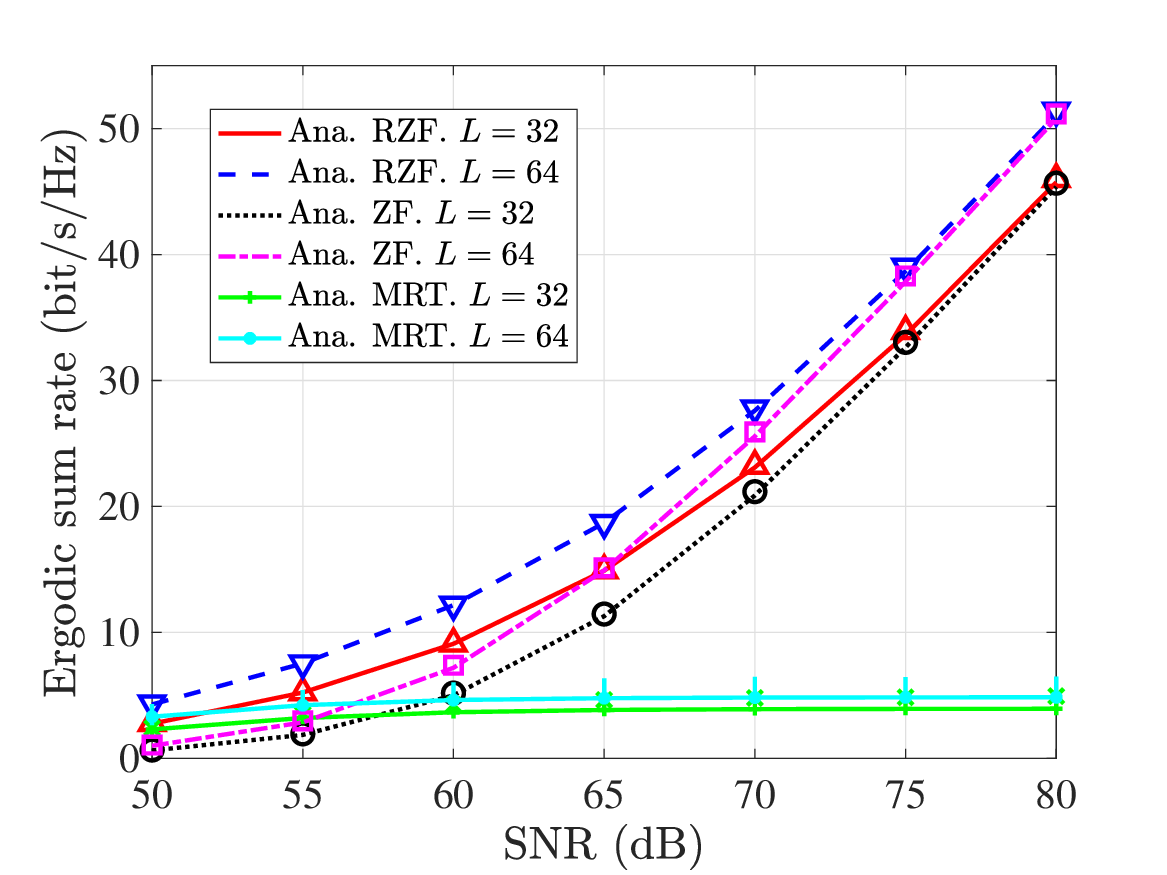}
\captionof{figure}{Comparison with ZF and MRT.}
\label{ZF_MRT_comp}
\vspace{-0.4cm}
\end{minipage}
\begin{minipage}[t]{0.33\textwidth}
\centering\includegraphics[width=1\textwidth]{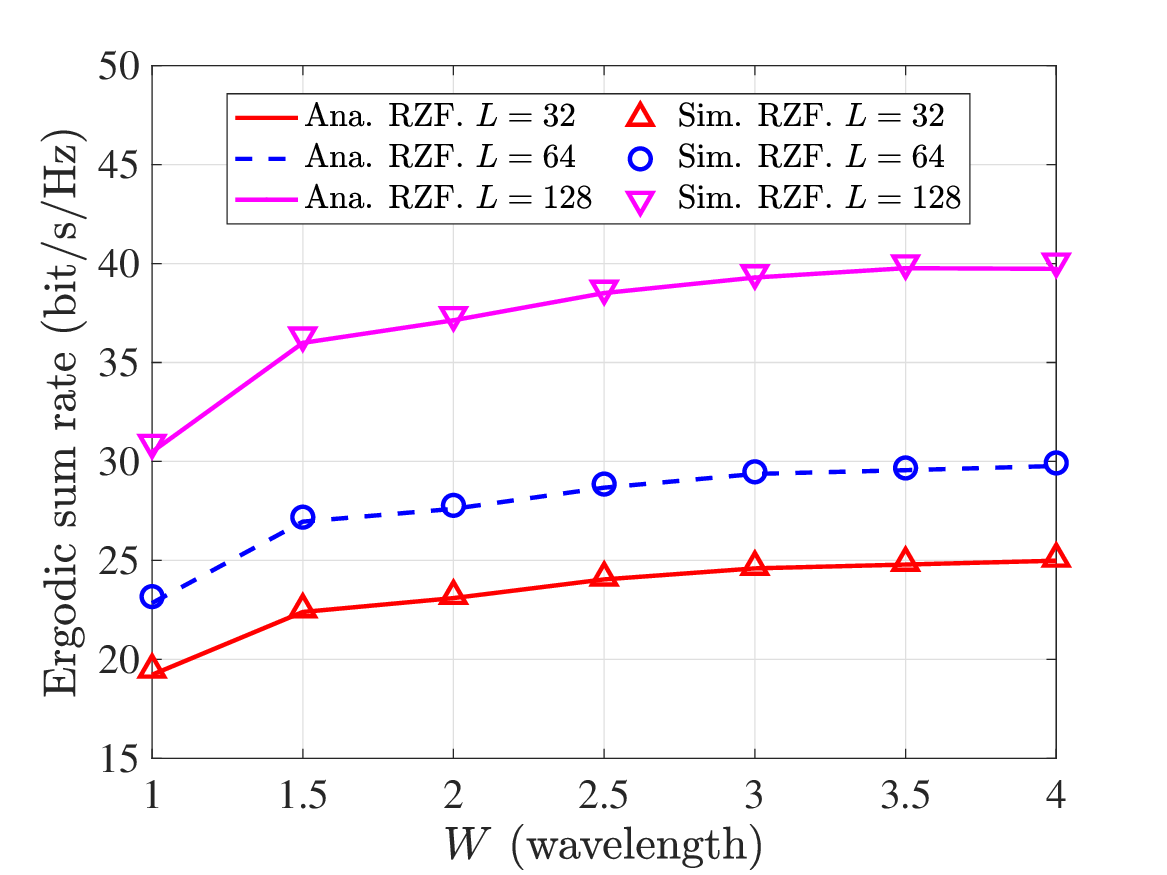}
\captionof{figure}{Impact of array size.}
\label{fig_impact_W}
\vspace{-0.4cm}
\end{minipage}
\begin{minipage}[t]{0.33\textwidth}
\centering\includegraphics[width=1\textwidth]{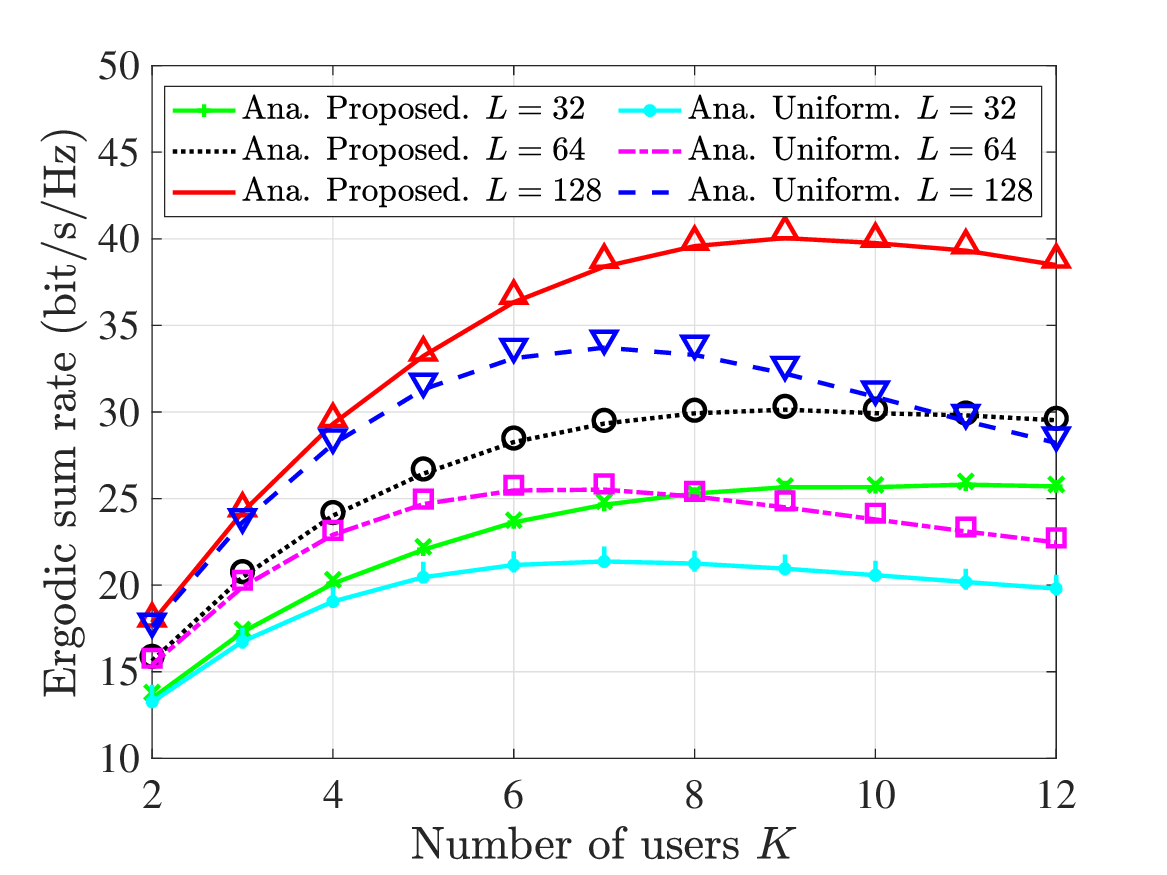}
\captionof{figure}{Impact of number of users.}
\label{fig_impact_users}
\vspace{-0.4cm}
\end{minipage}

\end{figure*}

\begin{figure}
\centering\includegraphics[width=0.33\textwidth]{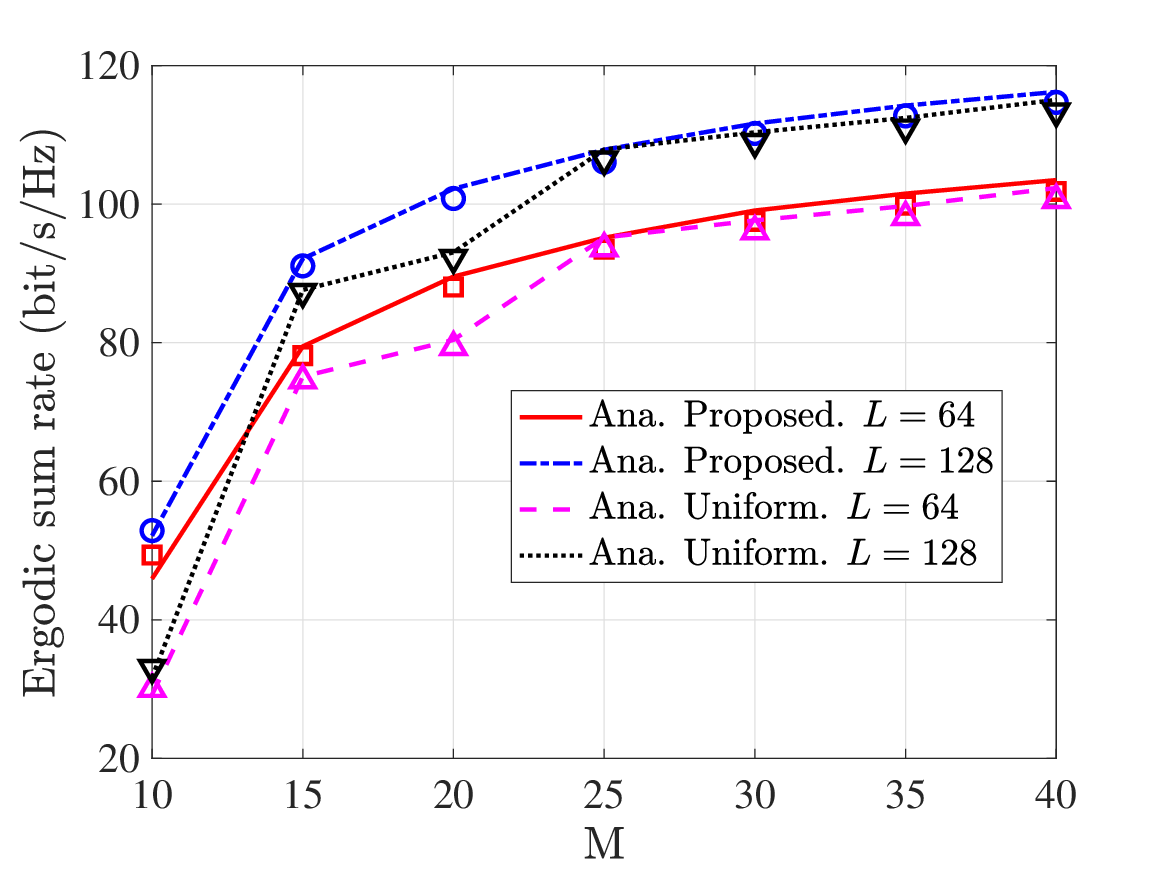}
\captionof{figure}{Impact of $M$.}
\label{rev_impact_port}
\vspace{-0.4cm}
\end{figure}
\begin{figure}
\centering\includegraphics[width=0.33\textwidth]{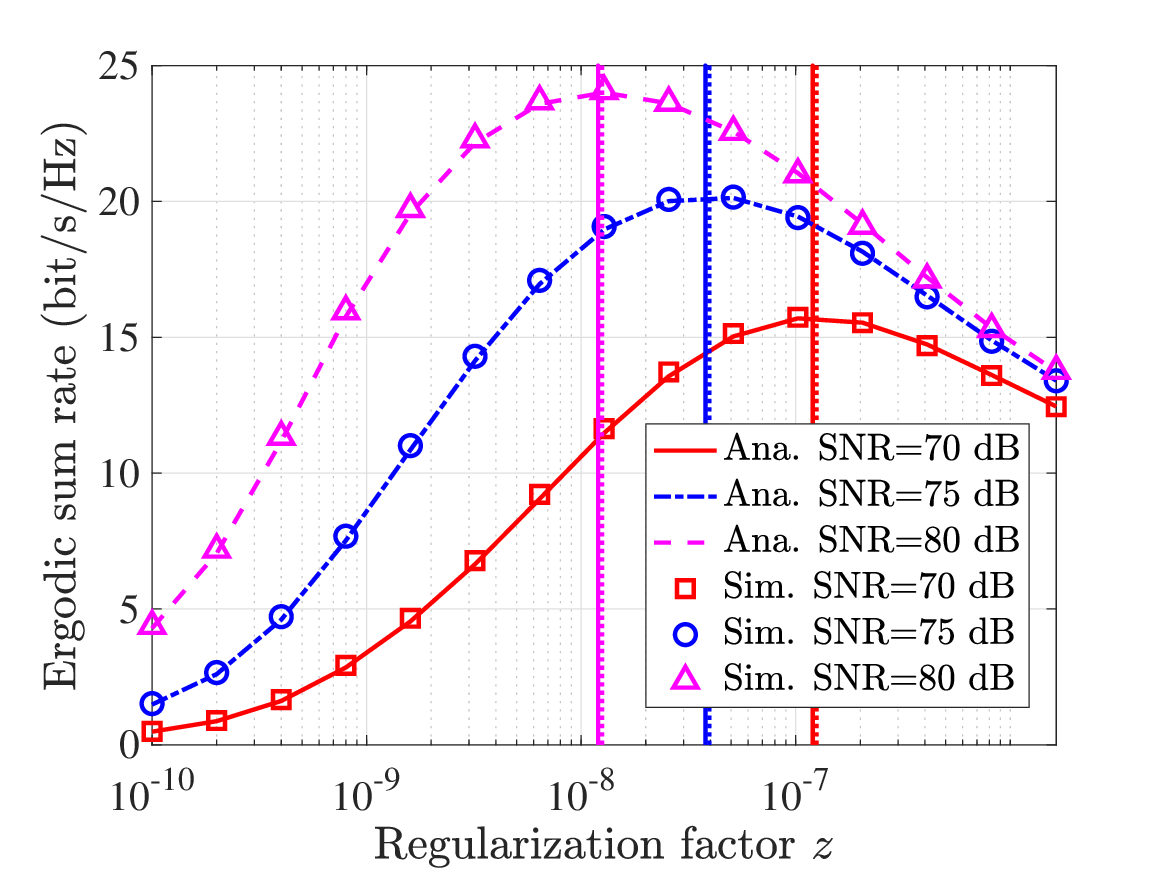}
\captionof{figure}{Optimal $z$ of RZF with homogeneous channels.}
\label{fig_opt_z}
\vspace{-0.4cm}
\end{figure}

\subsection{Performance of the Optimization Algorithm}
\label{simu_AO}

The performance of proposed Algorithm~\ref{AO_alg} is validated for the common scenario. The parameters are set to $M=20$, $K=8$, $L=32$, $d_{\mathrm{BS-RIS}}=5$ m, $d_{\mathrm{RIS-}k}=20+ \lfloor \frac{k-1}{4} \rfloor $ m, $k=1,2,...,K$. The angle between the BS-RIS and RIS-user link is $150^{\circ}$. Considering 3D environment under rich scattering, the correlation matrix is generated according to~\cite[Eq. (1)]{new2023information} such that the correlation coefficient between the $i$-th and $j$-th ports is given by $[\FR_{\mathrm{tot}}]_{i,j}=J_0\Bigl(2 \pi\sqrt{(\frac{|x_i -x_j|W_x}{N_x-1})^2+(\frac{|y_i -y_j|W_y}{N_y-1})^2 }\Bigl)$,
where $J_0(\cdot)$ is the spherical Bessel function of the first kind, and $(x_i,y_i)$ and $(x_j,y_j)$ are the coordinates of the $i$-th and $j$-th ports, respectively. Here we set $\BF_{\mathrm{tot}}=\FR_{\mathrm{tot}}$, $W_x=W_y=2\lambda$ and $N_x=N_y=10$ such that $M_{\mathrm{tot}}=100$. The other parameters are set as  $L=32$, $\FC_1=\BC(0.5,60,5,32)$, $\FC_2=\BC(0.5,30,5,32)$, $\FT=\diag(P_{\mathrm{RIS-}1},P_{\mathrm{RIS-}2},...,P_{\mathrm{RIS-}K})$, $\BU=\diag(P_{\mathrm{BS-}1},P_{\mathrm{BS-}2},...,P_{\mathrm{BS-}K})$, and $p_k=\lfloor \frac{k-1}{2} \rfloor +1$, $k=1,2,...,K$. The ports of the planar FAS are indexed increasingly from top to bottom, and left to right. Here we compare the proposed approach with uniform selection (sample uniformly according to indices, i.e.,  select the $1+(m-1)\lfloor \frac{M_{\mathrm{tot}}}{M-1} \rfloor$-th antenna with $m=1,2,...,M$). 

It can be observed from Fig.~\ref{fig_opt} that the proposed algorithm outperforms the uniform scheme, which demonstrates the gain obtained by the port selection and validates the effectiveness of Algorithm~\ref{overall_alg}. Fig.~\ref{ZF_MRT_comp} compares the ESR with RZF, ZF, and MRT. It can be observed that the RZF achieves the best performance while MRT can hardly tackle the interference among users. Moreover, the performance gap between RZF and ZF decreases as the SNR increases but the performance of MRT does not increase with the SNR. With the same setting, Fig.~\ref{fig_impact_W} depicts the ESR with increased aperture of FASs. It can be observed that when $W=W_x=W_y$ increases with fixed $M_{\mathrm{tot}}$, i.e., the spacing among ports increases and leads to low correlation among ports, the ESR saturates. The limiting performance is determined by the number of selected ports.

Fig.~\ref{fig_impact_users} depicts the ESR with different numbers of users, where $\FT=P_{\mathrm{RIS}-1}\BI_{K}$ and $\BU=P_{\mathrm{BS}-1}\BI_{K}$ with $d_{\mathrm{RIS}-1}=20$ m, $d_{\mathrm{BS}-1}=22.9$ m, and $\bold{p}=\bm{1}$. It can be observed that when the number of users is small, the ESR with uniform selection is close to that with the port selection obtained by Algorithm~\ref{overall_alg}. This indicates that when number of users is small, the performance is dominated by the precoding instead of the port selection, but when the number of users increases, the performance gain brought by port selection is more significant.

Fig.~\ref{rev_impact_port} depicts the ESR with different numbers of FAs, where $\sigma^{-2}=90$ dB. It can be observed that as $M$ increases, the ESR performance of uniform selection approaches that of proposed scheme, highlighting that the the performance gain of the proposed scheme over the uniform scheme with a smaller number of FAs is more significant than that with a larger number of FAs. This indicates that uniform selection is a good scheme when the number of FAs is large.

\subsection{Optimal Regularization Factor}

Fig.~\ref{fig_opt_z} shows the ESR with different values of $z$ over homogeneous channels, where $\FT=P_{\mathrm{RIS-}1}\BI_{K}$ and $\BU=P_{\mathrm{BS-}1}\BI_{K}$ with $d_{\mathrm{RIS}-1}=20, d_{\mathrm{BS}-1}=22.9$ m. The correlation matrices are same as that in Section~\ref{simu_AO} and $M=20$, $L=32$, $K=24$.  The theoretical optimal $z$ given in Proposition~\ref{pro_opt_z}, i.e., $z=\frac{K\sigma^2}{M}$, is plotted by the vertical line and $z$ obtained by $1$-dimension search is plotted by the vertical dotted line. It can be observed from Fig.~\ref{fig_opt_z} that the optimal $z$ given in Proposition~\ref{pro_opt_z} is accurate for estimating the optimal $z$ and thus the optimization for $z$ is unnecessary for homogeneous channels.

\section{Conclusion}
\label{sec_con}
In this paper, we investigated the two-timescale design for FAS-RIS MISO systems with RZF/ZF precoding. For that purpose, we first derived a closed-form deterministic equivalent for the SINR and per-user communication rate by leveraging RMT, for both the uncommon and common correlation cases. Based on the closed-form evaluation, we proposed a two-timescale design to maximize the ESR by jointly optimizing port selection, regularization factor of RZF, and phase shifts at the RIS. The results in this work can be utilized to determine the optimal regularization factor for RZF over homogeneous channels and the number of selected ports required to achieve a given ESR. Numerical results validate the accuracy of performance evaluation and show that the performance gain brought by port selection is more significant in FAS-RIS systems with large number than that with small number of users. Extending the two-timescale design to fit the imperfect CSI case is a promising future direction.

\appendices

\section{Proof of Theorem~\ref{uncommon_the}} 
\label{uncommon_the_proof}
 \begin{proof}
Before we start the proof, we first introduce the resolvent matrix of $\BH\BH^{H}$ and resolvent identity
 as
  \begin{equation}
  \label{def_res}
  \BQ=\left(z\bold{I}_{M}+\BH\BH^{H} \right)^{-1},~ \BQ\Bh_k=\frac{\BQ_k\Bh_k}{1+\Bh_k^{H}\BQ_k\Bh_k},
\end{equation}
 where $\BQ_k=\left(z\bold{I}_{M}+\BH_k\BH_k^{H} \right)^{-1}$ and $\BH_k$ can be obtained by removing the $k$-th column from $\BH$. Next, we will prove Theorem~\ref{uncommon_the}. By~(\ref{def_res}), we have
  \begin{align*}
&  \gamma_{\mathrm{RZF},k} =\frac{p_k  (\Bh^{H}_k\BQ\Bh_k)^2 }{  \Bh^{H}_k\BQ\BH_k\BP_k\BH^{H}_k \BQ\Bh_k + \sigma^2 \xi_{\mathrm{RZF}}}
  \\
  &
  =\frac{p_k  (\Bh^{H}_k\BQ_k\Bh_k)^2 }{  \Bh^{H}_k\BQ_k\BH_k\BP_k\BH^{H}_k \BQ_k\Bh_k + \sigma^2 \xi_{\mathrm{RZF}}(1+\Bh^{H}_k\BQ_k\Bh_k)^2}
  \\
  &
  =\frac{p_k A_k^2}{ B_k+\sigma^2 \xi_{\mathrm{RZF}}(1+\Bh^{H}_k\BQ_k\Bh_k)^2},    \numberthis   \label{simple_SINR}
  \end{align*}
which indicates that the evaluation of $\gamma_{\mathrm{RZF},k}$ can resort to that of $A_k$, $B_k$, and $\xi_{\mathrm{RZF}}$. For the simplicity of notations, we denote $K \xlongrightarrow{(c_1,c_2)} \infty$ by $(K)_{\infty}$. By the independence between $\Bh_k$ and $\BQ_k$, , we can obtain
\begin{align*}
&   A_k  \xlongrightarrow[(K)_{\infty}]{\mathcal{P}}\frac{1}{M}\E[\Tr(\FF_k\BQ_k)]+ \frac{1}{L}\E[\Tr(\BZ_k\BZ^{H}_k\BQ_k)],
\\
&
 B_k \xlongrightarrow[(K)_{\infty}]{\mathcal{P}}\frac{1}{M}\E[\Tr(\FF_k\BQ_k\BH_k\BP_k\BH_k^{H}\BQ_k)]
 \\
 &
 + \frac{1}{L}\E[\Tr(\BZ_k\BZ^{H}_k\BQ_k\BH_k\BP_k\BH_k^{H}\BQ_k)] \xlongrightarrow[(K)_{\infty}]{\mathcal{P}}
 \\
 &  \sum_{l \neq k} \frac{p_l}{(1+\E[\Bxi^{H}_l\BQ_{k,l}\Bxi_l])^2} (\frac{1}{M^2}\E[\Tr(\FF_k\BQ_{k,l}\FF_l\BQ_{k,l})]
 \\
 &
 + \frac{1}{ML}(\E[\Tr(\BZ_k\BZ^{H}_k\BQ_{k,l}\FF_l \BQ_{k,l})]
 +\E[\Tr(\FF_k\BQ_{k,l}\BZ_l\BZ_l^{H}\BQ_{k,l})])
 \\
 &
 + \frac{1}{L^2}\E[\Tr(\BZ_k\BZ^{H}_{k}\BQ_{k,l}\BZ_l\BZ_l^{H} \BQ_{k,l})]
 )
 \\
&
\xi_{\mathrm{RZF}} \xlongrightarrow[(K)_{\infty}]{\mathcal{P}} \frac{1}{M}\E[\Tr(\BQ^2\BH\BP\BH^{H})]\xlongrightarrow[(K)_{\infty}]{\mathcal{P}}
\\
&
\sum_{l =1}^{K} \frac{p_l(\frac{1}{M}\E[\Tr(\FF_l\BQ_l^2)]+ \frac{1}{L}\E[\Tr(\BZ_l\BZ^{H}_l\BQ_l^2)])}{M(1+\E[\Bh^{H}_l\BQ_{l}\Bh_l])^2},
 \numberthis \label{first_SINR}
\end{align*}
Next, we first investigate the convergence for the trace of the resolvent and then obtain the deterministic equivalent for SINR and per-user rate using continuous mapping theorem. The evaluation of the $A_k$, $B_k$, and $\xi_{\mathrm{RZF}}$ can be obtained by the following two lemmas.
\begin{lemma}
\label{trace_appro} (First-order resolvent results) Given assumptions~\textbf{A.1}-\textbf{A.3}, for the random matrix $\BH=[\Bh_1,...,\Bh_K]$ with $\Bh_k$ defined in~(\ref{chunc}), there holds true that
\begin{align*}
\label{trace_exp_con}
& \frac{1}{M}\E[\Tr(\FR\BQ)] \xrightarrow[]{(K)_{\infty}\rightarrow \infty}  \delta,~\frac{1}{L}\E[\Tr(\BZ_k\BZ_k^{H}\BQ)]\xrightarrow[]{(K)_{\infty}}  
\omega_k,
\\
&  \frac{1}{M}\E[\Tr(\BF_k\BQ)]+\frac{1}{L}\E[\Tr(\BZ_k\BZ_k^{H}\BQ)] \xrightarrow[]{(K)_{\infty}}\mu_k,
 \numberthis
\end{align*}
where $(\delta,\bm{\mu})$ is the unique positive solution for the system of equations in~(\ref{basic_eq1}). 
\end{lemma}
\begin{proof} Lemma~\ref{trace_appro} can be proved by Gaussian tools, which is omitted here.  
\end{proof}
 
\begin{lemma}\label{lem_sec_resol} (Second-order resolvent results) Given assumptions~\textbf{A.1}-\textbf{A.3} and $\| \BK\|<\infty$, there holds true
\begin{align*}
&\Upsilon_k(\BK)=\frac{1}{L}\E[\Tr(\BZ_k\BZ^{H}_{k}\BQ\BK \BQ)]+\frac{1}{M}\E[\Tr(\BF_k\BQ\BK \BQ)]
\\
&
\xrightarrow[]{(K)_{\infty}}[\bm{\Pi}^{-1} \bm{\chi}(\BK)]_{k},
\\
&\Lambda_{k,l}=\frac{1}{L}\E[\Tr(\BZ_k\BZ^{H}_{k}\BQ\BZ_l\BZ_l^{H} \BQ)]+\frac{1}{M}\E[\Tr(\BZ_k\BZ^{H}_{k}\BQ\BF_l \BQ)]
\\
&
\xrightarrow[]{(K)_{\infty}}[\bm{\Delta}^{-1}\Bxi_l]_k+\frac{L}{M\delta^2}[\bm{\Delta}^{-1}\Bxi_I]_k [\bm{\Pi}^{-1}\bm{\chi}(\FR)]_l   
\\
&
+\frac{L}{M}[\bm{\Pi}^{-1}\bm{\chi}(\overline{\BF}_k)]_{l}.
\numberthis \label{sec_resol}
\end{align*}
\end{lemma}
\begin{proof} The proof of Lemma~\ref{lem_sec_resol} is given in Appendix~\ref{pro_lem_sec_resol}.
\end{proof}
By~(\ref{first_SINR}), rank-one perturbation lemma~\cite[Lemma 14.3]{couillet2011random}, and Lemma~\ref{trace_appro}, we can obtain 
\begin{align*}
&A_k \xrightarrow[]{(K)_{\infty}}\frac{1}{M}\E[\Tr(\FF_k\BQ)]
+ \frac{1}{L}\E[\Tr(\BZ_k\BZ^{H}_k\BQ)] \xrightarrow[]{(K)_{\infty}} \mu_k.\numberthis \label{de_Ak}
\end{align*}
By~(\ref{first_SINR}),~\cite[Lemma 14.3]{couillet2011random}, and Lemma~\ref{lem_sec_resol}, we can obtain 
\begin{align*}
& B_k \xlongrightarrow[(K)_{\infty}]{\mathcal{P}}
 \sum_{l \neq k} \frac{p_l}{L(1+\mu_l)^2}(\Lambda_{k,l}+\frac{L}{M}\Upsilon_l(\BF_k) )
 \\
 &
\xlongrightarrow[(K)_{\infty}]{\mathcal{P}} \sum\limits_{l\neq k}^{K} \frac{p_l\Psi_{k,l}}{L(1+\mu_l)^2},
\\
&\xi_{\mathrm{RZF}} \xlongrightarrow[(K)_{\infty}]{\mathcal{P}} \sum_{l =1}^{K} \frac{p_l \Upsilon_l(\BI_M) }{M(1+\mu_l)^2}
 \\
 &
 \xlongrightarrow[(K)_{\infty}]{\mathcal{P}}
\sum_{l =1}^{K} \frac{p_l [\bm{\Pi}^{-1}\bm{\chi}(\BI_M)]_l }{M(1+\mu_l)^2}=\overline{C}.\numberthis \label{de_Bk}
\end{align*}
By plugging~(\ref{de_Ak}) and~(\ref{de_Bk}) into~(\ref{simple_SINR}), we can conclude~(\ref{SINR_de}). Then, by applying the continuous mapping theorem and
the dominated convergence theorem~\cite{billingsley2017probability}, we have can obtain $\frac{R_{\mathrm{RZF}}}{K}   \xlongrightarrow[(K)_{\infty}]{\mathcal{P}}  \frac{\overline{R}_{\mathrm{RZF}}}{K}$. The convergence $\frac{\E[R_{\mathrm{RZF}}]}{K}   \xlongrightarrow[(K)_{\infty}]{\mathcal{P}}   \frac{\overline{R}_{\mathrm{RZF}}}{K}$ can be obtained similarly by the convergence of $ \E [\gamma_{\mathrm{RZF},k}]\xlongrightarrow[]{(K)_{\infty}} \overline{\gamma}_{\mathrm{RZF},k}$. 
\end{proof}

\section{Proof of Theorem~\ref{the_ZF}}
\label{proof_the_ZF}
\begin{proof} Denote $\widetilde{\delta}=\lim\limits_{z\rightarrow 0}z\delta(z)$, $\widetilde{\omega}_k =\lim\limits_{z\rightarrow 0}z\omega_k(z)$, and $\widetilde{\mu}_k =\lim\limits_{z\rightarrow 0}z\mu_k(z)$. Based on definition of $\delta$ in~(\ref{basic_eq1}), we have
\begin{equation}
\begin{aligned}
&\widetilde{\delta}=\lim\limits_{z\rightarrow 0}z\delta(z)=\lim\limits_{z\rightarrow 0} \frac{1}{M}\Tr\Bigl(\FR
\\
&\times \left( \bold{I}_{M}  + \sum_{k=1}^{K} \frac{\FF_k}{M(z+z\mu_k(z))}+ \frac{z \omega_k(z)\FR}{Mz\delta(z)(z+z\mu_k(z))}  \right)^{-1}\Bigr)
\\
&
=\frac{1}{M}\Tr\Bigl(\FR\left( \bold{I}_{M} + \sum_{k=1}^{K}\frac{\widetilde{\omega}_k\FR}{M\widetilde{\delta}\widetilde{\mu}_k}+ \frac{\BF_k}{M\widetilde{\mu}_k}  \right)^{-1}\Bigr)=\frac{1}{M}\Tr(\bold{R}\widetilde{\BK}_R),
\end{aligned}
\end{equation}
where $\widetilde{\BK}_R=\left( \bold{I}_{M} + \sum_{k=1}^{K}\frac{\widetilde{\omega}_k\FR}{M\widetilde{\delta}\widetilde{\mu}_k}+ \frac{\BF_k}{M\widetilde{\mu}_k}  \right)^{-1}$. Moreover, we have 
\begin{equation}
\begin{aligned}
\widetilde{\omega}_k &=\lim\limits_{z\rightarrow 0}z\omega_k(z)=\frac{1}{L}\Tr( \FC_k\widetilde{\BK}_{C}),
\\
\widetilde{\mu}_k &=\lim\limits_{z\rightarrow 0}z\mu_k(z)=\widetilde{\omega}_k+\frac{1}{M}\Tr(\BF_k\widetilde{\BK}_R),
\end{aligned}
\end{equation}
where $\widetilde{\BK}_C=\left( \frac{1}{\widetilde{\delta}}\bold{I}_{L}+\sum_{k=1}^{K}\frac{\FC_k}{L\widetilde{\mu}_k}   \right)^{-1}$. This indicates that
showing that $\widetilde{\delta}$, $\widetilde{\omega}_k$, and $\widetilde{\mu}_k$ can be obtained by solving~(\ref{basic_zf1}). Thus, we can obtain 
\begin{equation}
\begin{aligned}
\underline{\delta}&=\widetilde{\delta}=\lim\limits_{z\rightarrow 0}z\delta(z),~\underline{\omega}_k= \widetilde{\omega}_k =\lim\limits_{z\rightarrow 0}z\omega_k(z),
\\
\underline{\mu}_k &=\widetilde{\mu}_k =\lim\limits_{z\rightarrow 0}z\mu_k(z),
\end{aligned}
\end{equation}
 Moreover, observing that $\lim\limits_{z\rightarrow 0} z^2 \frac{- \mathrm{d} \mu_k}{\mathrm{d} z} = \underline{\mu}_k$, we can obtain 
\begin{equation}
\lim\limits_{z\rightarrow 0} \overline{C}=\lim\limits_{z\rightarrow 0} \frac{1}{M}\sum_{k=1}^{K}\frac{ p_k z^2 \frac{- \mathrm{d} \mu_k}{\mathrm{d} z}  }{(z+z\mu_k)^2}=\sum_{k=1}^{K}\frac{p_k}{M \underline{\mu}_k}.
\end{equation}
Noticing that the interference term is cancelled by ZF precoding, according to~(\ref{SINR_de}), we have
\begin{equation}
\begin{aligned}
\lim\limits_{z\rightarrow 0}  \overline{\gamma}_{\mathrm{RZF},k} &=\lim\limits_{z\rightarrow 0} \frac{p_k z^2 \mu_k^2}{  \sigma^2(z+z\mu_k)^2\overline{C} }
\\
&=p_k(\sigma^2\sum\limits_{l=1}^{K}\frac{p_l}{M\underline{\mu}_l})^{-1}=\overline{\gamma}_{\mathrm{ZF},k} ,
\end{aligned}
\end{equation}
and
\begin{equation}
\begin{aligned}
\gamma_{\mathrm{ZF},k} =\lim\limits_{z\rightarrow 0} \gamma_{\mathrm{RZF},k}  \xlongrightarrow[K \xlongrightarrow{(c_1,c_2)} \infty]{\mathcal{P}}\lim\limits_{z\rightarrow 0}  \overline{\gamma}_{\mathrm{RZF},k} =\overline{\gamma}_{\mathrm{ZF},k},
\end{aligned}
\end{equation}
which concludes Theorem~\ref{the_ZF}.
\end{proof}

\section{Proof of Proposition~\ref{the_ESR_common}} 
\label{proof_sinr_com}
\begin{proof} Similar to the proof of Theorem~\ref{uncommon_the}, the proof of Proposition~\ref{the_ESR_common} relies on the following second-order resolvent results, which can be proved by the similar approach as Lemma~\ref{lem_sec_resol}.  
\begin{lemma}\label{lem_sec_resol_com} (Second-order resolvent results for common correlation case) Given assumptions~\textbf{A.1}-\textbf{A.3} and $\| \BK\|<\infty$, there holds true
\begin{align*}
& \frac{1}{M}\E[\Tr(\FR\BQ)] \xrightarrow[]{(K)_{\infty}\rightarrow \infty}  \delta,~\frac{1}{L}\E[\Tr(\BZ\BZ^{H}\BQ)]\xrightarrow[]{(K)_{\infty}}  
\omega,\\
&  \frac{1}{M}\E[\Tr(\BF\BQ)] \xrightarrow[]{(K)_{\infty}}\kappa,
\\
&\Lambda=\frac{1}{L}\E[\Tr(\BZ\BZ^{H}\BQ\BZ\BZ^{H} \BQ)]\xrightarrow[]{(K)_{\infty}} 
\\
&\frac{1}{\Delta}(\Xi+ \frac{L}{M}\Xi\eta(\FT,\BU)[\bm{\Pi}_{\mathrm{com}}^{-1}\bm{\chi}(\BF)]_3+\frac{L\Xi_I }{M\delta^2}[\bm{\Pi}_{\mathrm{com}}^{-1}\bm{\chi}(\FR)]_3 ),
\\
&\Gamma(\BK)=\frac{1}{M}\E[\Tr(\BF\BQ\BK \BQ)] \xrightarrow[]{(K)_{\infty}}[\bm{\Pi}_{\mathrm{com}}^{-1}\bm{\chi}(\BK)]_2,
\\
&
\Theta(\BK)=\frac{1}{L}\Tr(\BZ\BZ^{H}\BQ\BK\BQ) \xrightarrow[]{(K)_{\infty}}[\bm{\Pi}_{\mathrm{com}}^{-1}\bm{\chi}(\BK)]_3.
\numberthis \label{sec_resol_com}
\end{align*}
\end{lemma}
Based on the decomposition in~(\ref{first_SINR}), we first obtain the following evaluation for $A_k$, $B_k$, and $\xi_{\mathrm{RZF}}$ in~(\ref{simple_SINR}) as 
\begin{align*}
&   A_k  \xlongrightarrow[(K)_{\infty}]{\mathcal{P}} t_k\omega+u_k\kappa,
\\
&
 B_k\xlongrightarrow[(K)_{\infty}]{\mathcal{P}} \sum_{l\neq k}\frac{p_l ( t_k   t_l \Lambda+ \frac{L(t_k u_l + t_l u_k)}{M}\Theta(\BF)+\frac{Lu_k u_l}{M} \Gamma(\BF)  )  ) }{L(1+t_l\omega+u_l\kappa)^2}
 \\
 &
   \xlongrightarrow[(K)_{\infty}]{\mathcal{P}} \sum_{l\neq k}\frac{p_l\Psi_{k,l}}{L(1+t_l\omega+u_l\kappa)^2},
 \\
 &
\xi_{\mathrm{RZF}} \xlongrightarrow[(K)_{\infty}]{\mathcal{P}}\sum_k  \frac{p_k  (\Theta(\BI_M)+ \Gamma(\BI_M)) }{M(1+t_k\omega+u_k\kappa)^2}   \xlongrightarrow[(K)_{\infty}]{\mathcal{P}} \overline{C}_{\mathrm{com}}.
 \numberthis \label{first_SINR_com}
\end{align*}
Thus, we can conclude~(\ref{SINR_com_exp}) and the rest of the proof is the same as Appendix~\ref{uncommon_the_proof}.
\end{proof}

\section{Proof of Proposition~\ref{the_ZF_common}}
\label{proof_the_ZF_common}
\begin{proof} Denote $\widetilde{\delta}=\lim\limits_{z\rightarrow 0}z\delta(z)$, $\widetilde{\omega} =\lim\limits_{z\rightarrow 0}z\omega(z)$, and $\widetilde{\kappa} =\lim\limits_{z\rightarrow 0}z\kappa(z)$. Thus, we have 
\begin{equation}
\begin{aligned}
\widetilde{\overline{\omega}}=\lim\limits_{z\rightarrow 0}z^{-1}\overline{\omega}=\frac{1}{L}\Tr(\FT\widetilde{\BPS}_{T}), 
\\
\widetilde{\overline{\kappa}}=\lim\limits_{z\rightarrow 0}z^{-1}\overline{\kappa}=\frac{1}{L}\Tr(\BU\widetilde{\BPS}_{T}),
\end{aligned}
\end{equation}
where $\widetilde{\BPS}_{T} =\left(\widetilde{\kappa}\BU+\widetilde{\omega}\bold{T}\right)^{-1}$.
Based on definition of $\delta$ in~(\ref{basic_eq2}), we have
\begin{align*}
&\widetilde{\delta}=\lim\limits_{z\rightarrow 0}z\delta(z)=\lim\limits_{z\rightarrow 0} \frac{1}{M}\Tr\Bigl(\FR  \numberthis
\\
&\times \left( \bold{I}_{M}  + \frac{\FF}{M(z+z\mu(z))}+ \frac{z \omega(z)\FR}{Mz\delta(z)(z+z\mu_k(z))}  \right)^{-1}\Bigr)
\\
&
=\frac{1}{M}\Tr\Bigl(\FR\left( \bold{I}_{M}+\frac{L z\omega z^{-1}\overline{\omega}}{M z\delta}\bold{R} + \frac{L z^{-1}\overline{\kappa}}{M}\BF\right)^{-1}\Bigr)
\\
&
=\frac{1}{M}\Tr(\bold{R}\widetilde{\BK}_R),
\end{align*}
where $\widetilde{\BPS}_R=\left( \bold{I}_{M}+ \frac{L\widetilde{\overline{\kappa}}}{M}\BF+\frac{L \widetilde{\omega}\widetilde{\overline{\omega}}}{M\widetilde{\delta}}\bold{R}   \right)^{-1}$. Moreover, we have 
\begin{equation}
\begin{aligned}
\underline{\delta}&=\widetilde{\delta}=\lim\limits_{z\rightarrow 0}z\delta(z),~\underline{\kappa}=\widetilde{\kappa} =\lim\limits_{z\rightarrow 0}z\kappa(z),
\\
\underline{\omega}&=\widetilde{\omega} =\lim\limits_{z\rightarrow 0}z\omega(z)=\frac{1}{L}\Tr(\BC\widetilde{\BPS}_C),
\end{aligned}
\end{equation}
where $ \underline{\BPS}_{C}=\left(\frac{1}{\widetilde{\delta}}\bold{I}_{L}+\widetilde{\overline{\omega}}\BC\right)^{-1}$. This indicates that $\widetilde{\delta}$, $\widetilde{\omega}$, $\widetilde{\kappa}$, $\widetilde{\overline{\omega}}$, and $\widetilde{\overline{\kappa}}$ can be obtained by solving~(\ref{basic_zfeq2}). Thus, we have
\begin{equation}
\begin{aligned}
\underline{\delta}&=\widetilde{\delta}=\lim\limits_{z\rightarrow 0}z\delta(z),~\underline{\omega}= \widetilde{\omega} =\lim\limits_{z\rightarrow 0}z\omega(z),
\\
\underline{\kappa} &=\widetilde{\kappa} =\lim\limits_{z\rightarrow 0}z\kappa(z),
\end{aligned}
\end{equation}
Moreover, observing that $\lim\limits_{z\rightarrow 0} z^2 \frac{- \mathrm{d} (t_k\omega+u_k\kappa)}{\mathrm{d} z} = t_k\underline{\omega}+u_k\underline{\kappa}$, we can obtain 
\begin{equation}
\begin{aligned}
&\lim\limits_{z\rightarrow 0} \overline{C}_{\mathrm{com}}=\lim\limits_{z\rightarrow 0} \frac{1}{M}\sum_{k=1}^{K}\frac{ p_k z^2 \frac{- \mathrm{d} (t_k\omega+u_k\kappa)}{\mathrm{d} z}  }{(z+z( t_k\omega+u_k\kappa))^2}
\\
&
=\sum_{k=1}^{K}\frac{p_k}{M ( t_k\underline{\omega}+u_k\underline{\kappa})}.
\end{aligned}
\end{equation}
Noticing that the interference term is cancelled by ZF precoding, according to~(\ref{SINR_com_exp}), we can obtain
\begin{equation}
\begin{aligned}
\lim\limits_{z\rightarrow 0}  \overline{\gamma}_{\mathrm{RZF},k} &=\lim\limits_{z\rightarrow 0} \frac{p_k z^2 (t_k\omega+u_k\kappa)^2}{  \sigma^2(z+z(t_k\omega+u_k\kappa))^2\overline{C}_{\mathrm{com}} }
\\
&=p_k(\sigma^2\sum\limits_{l=1}^{K}\frac{p_l}{M(t_l\omega+u_l\kappa)})^{-1}=\overline{\gamma}_{\mathrm{ZF},k} ,
\end{aligned}
\end{equation}
and
\begin{equation}
\begin{aligned}
\gamma_{\mathrm{ZF},k} =\lim\limits_{z\rightarrow 0} \gamma_{\mathrm{RZF},k}  \xlongrightarrow[K \xlongrightarrow{(c_1,c_2)} \infty]{\mathcal{P}}\lim\limits_{z\rightarrow 0}  \overline{\gamma}_{\mathrm{RZF},k} =\overline{\gamma}_{\mathrm{ZF},k},
\end{aligned}
\end{equation}
which concludes Proposition~\ref{the_ZF_common}.
\end{proof}

\section{Proof of Proposition~\ref{iid_pro}}
\label{proof_iid_pro}
\begin{proof} By setting $\FR=\FF=\BI_M$, $\FC=\BI_L$, $\BU=u\BI_K$, and $\FT=t\BI_K$ in Proposition~\ref{the_ESR_common}, the system of equation in~(\ref{basic_zfeq2}) degenerates to 
 \begin{equation}
 \label{basic_iid2}
 \begin{aligned}
 \begin{cases}
   &\underline{\delta}-(1+\frac{c_1t\underline{\omega}}{\underline{\mu}\underline{\delta}}+\frac{c_1 u}{\underline{\mu}} )^{-1}=0,
   \\
   & \underline{\omega}-(\frac{1}{\underline{\delta}}+\frac{c_2t}{\underline{\mu}}  )^{-1}=0,
   \\
 &\underline{\mu}-t\underline{\omega}-u\underline{\delta}=0.
\end{cases}
  \end{aligned}
 \end{equation}
 The solution of~(\ref{basic_iid2}) is $\underline{\mu}=(1-c_1)\beta(u,t,c_2)$. By~(\ref{SINR_ZF_com}), we have
 \begin{equation}
\overline{\gamma}_{\mathrm{ZF},k}= (\sigma^2\frac{K}{M\underline{\mu}})^{-1}=\frac{(1-c_1)\beta(u,t,c_2)}{c_1\sigma^2},
\end{equation}
which concludes the proof.
\end{proof}
\section{Proof of Proposition~\ref{pro_opt_z}}
\label{proof_pro_opt_z}
\begin{proof}
Proposition~\ref{pro_opt_z} can be proved by letting $\frac{\mathrm{d} R_{\mathrm{RZF}}}{\mathrm{d}z}=0$. With the homogeneous setting, we have $\Psi_{i,j}=\Psi$, $i,j=1,...,K$. Denote $\mu=t\omega+u\kappa$ and notice that
\begin{equation}
\mu=A-z\mu',
\end{equation}
with $A=\frac{(K-1)\Psi}{L(1+\mu)^2}$ and $\mu'=\frac{\mathrm{d} \mu}{\mathrm{d}z}$. We have $2\mu'=A'-z\mu''-\mu'+\mu'=A'-z\mu''$ with $\mu''=\frac{\mathrm{d}^2 \mu}{\mathrm{d} z^2}$, and
\begin{equation}
\begin{aligned}
\frac{\mathrm{d} \overline{\gamma}_{\mathrm{RZF}}}{\mathrm{d}z}
&=[(A'-z\mu'') (A- \frac{K\sigma^2}{M}\mu'  )
-(A-z\mu')
\\
&
\times (A'-\frac{K\sigma^2}{M}\mu''  )]
\mu(A- \frac{K\sigma^2}{M}\mu' )^{-2}
\\
&=(z-\frac{N\sigma^2}{M})(\mu'  A'-A\mu'').
\end{aligned}
\end{equation}
We can prove that $\mu'  A'-A\mu''\neq 0$ such that the solution of $\frac{\mathrm{d} R_{\mathrm{RZF}}}{\mathrm{d}z}=0$ is $z=\frac{N\sigma^2}{M}$, which concludes the proof.
\end{proof}

\section{Proof of Lemma~\ref{lem_sec_resol}}
\label{pro_lem_sec_resol}
\begin{proof} The proof of Lemma~\ref{lem_sec_resol} mainly relies on Gaussian tools, i.e., the \textit{Integration by Parts Formula}~\cite[Eq. (17)]{hachem2008new} and the~\textit{Nash-Poincar{\'e} Inequality}~\cite[Eq. (18)]{hachem2008new}, which are utilized for computation and error control, respectively. We first evaluate $\Upsilon(\BK)$. By the integration by parts formula, we have 
\begin{align*}
&\E [[\LC_l\By_l\Bh^{H}_l\BQ\BK\BQ\bold{A}\bold{X}\bold{B}]_{p,q}]

\\
&=\sum_{k=1}^{L}\E [ \frac{\alpha_{\overline{\mu},l}}{L} [  \LC_l \BZ_l^{H}\BQ\BK\BQ\BA\bold{X}\bold{B}]_{p,q}  
- [\frac{\alpha_{\overline{\mu},l}}{L}\Tr(\BZ_l\BZ^{H}_l\BQ\BK\BQ)
\\
&
+\frac{\alpha_{\overline{\mu},l}}{M}\Tr(\BF_l\BQ\BK\BQ)] [\LC_l\By_l\Bh_l^{H}\BQ\bold{A}\bold{X}\bold{B}]_{p,q}  ]\numberthis \label{hhQKQ}
\\
&
- \alpha_{\overline{\mu},l}\cov(\frac{1}{L}\Tr(\BZ_l\BZ^{H}_l\BQ)+\frac{1}{M}\Tr(\BF_l\BQ) 
\\
&
[\LC_l\By_l\Bh_l^{H}\BQ\BK\BQ\BA\bold{X}\BB]_{p,q}),
\end{align*}
where $\BA$ and $\BB$ are two deterministic matrices with bounded norm and $\alpha_{\overline{\mu},l}=(1+\E[\frac{1}{L}\Tr(\BZ_l\BZ^{H}_l\BQ)+\frac{1}{M}\Tr(\BF_l\BQ)])^{-1}=(1+\mu_l)^{-1}+o(1)$. 
By letting $p=q$, summing over $p$, and using Lemma~\ref{trace_appro}, we have
\begin{equation}
\label{eva_zeta_k}
\begin{aligned}
& \zeta_{k}(\BA,\BK)=\frac{1}{L}\E[\Tr(  \BZ^{H}_k\BQ\BK\BQ\BA\BX\LC_k)]=
\frac{\omega_k}{\delta}\Gamma(\BK,\BA\LR)
\\
&
+\sum_{l=1}^{K}\frac{\Xi_{k,l} }{L\delta(1+\mu_l)^2}\frac{\Tr(\BA\LR\BPS_R)}{M}\Upsilon_{l}(\BK)
\\
&-  \frac{\Tr(\BC_k\BPS_C\FC_l\BPS_C)}{\delta^2 L^2(1+\omega_l)}\frac{\Tr(\BA\LR\BPS_R)}{M} \Gamma(\BK,\FR)+o(1),
\end{aligned}
\end{equation}
where $\Gamma(\BL,\BK)=\frac{1}{M}\E[\Tr(\BL\BQ\BK\BQ)]$. Next, we will evaluate $\Gamma(\BL,\BK)$ and construct a system of equation with respect to $\Upsilon(\BK)$ and $\Gamma(\FR,\BK)$. By the identity $\BA-\BB=\BB(\BB^{-1}-\BA^{-1})\BA$, we have
\begin{equation}
\label{QKQL_ex}
\begin{aligned}
&\Gamma(\BL,\BK)=\frac{1}{M}\E[\Tr(\BK\BQ\BL(\BQ-\BPS_R))]
\\
&
+\frac{1}{M}\E[\Tr(\BK\BQ\BL\BPS_R)]
=
\frac{1}{M}\E[\Tr(\BK\BQ\BL\BPS_R)]+
\\
&
(\sum_{k=1}^{K}\frac{\omega_k\overline{\mu}_k}{M\delta}\E[\Tr(\BQ\BK\BQ\BL\BPS_R\FR)]+\frac{\overline{\mu}_k}{M}\E[\Tr(\BQ\BK\BQ\BL\BPS_R\BF_k)])
\\
&
-\frac{1}{M}\E[\Tr(\BK\BQ\BL\BPS_{R}\BH\BH^{H}\BQ)]
\\
&
=\frac{\Tr(\FR\BPS_{R}\BK\BPS_{R})}{M}+T_{1}-T_{2}+\BO(\frac{1}{M}),
\end{aligned}
\end{equation}
with $\overline{\mu}_k=(1+\mu_k)^{-1}$. By~(\ref{hhQKQ}) and~(\ref{eva_zeta_k}), we can evaluate $T_2$ by
\begin{align*}
T_{2}&=\frac{1}{M}\sum_{k=1}^{K} \Bh_k^{H}\BQ\BK\BQ\BL\BPS_R\Bh_k
=\sum_{k=1}^{K} [\frac{\omega_k}{M\delta} \Gamma(\BL\BPS_R\FR,\BK)
\\
&
+\frac{\overline{\mu}_k}{M}\Gamma(\BL\BPS_R\BF_k,\BK)] +[-\frac{\overline{\mu}_k^2}{M}  (\frac{\Xi_{I,k}}{\delta^2} \chi(\BL,\FR)+\chi(\BF_k,\BL))
\\
&\times
 \Upsilon_{k}(\BK) -\sum_l\Gamma(\BK,\FR)\frac{\overline{\mu}_l \overline{\mu}_k\Xi_{k,l}}{ML\delta^2}\chi(\BL,\FR)]+o(1)
\\
&
=R_{1}+R_{2}+o(1).\numberthis \label{T2exp}
\end{align*}
By plugging~(\ref{T2exp}) into~(\ref{QKQL_ex}) and noticing $T_{1}=R_{1}$, we can obtain $\Gamma(\BK,\BL)$
\begin{align*}
\label{QKQL}
&\Gamma(\BK,\BL)=\chi(\BK,\BL)+\sum_{k=1}^{K}\frac{  (\frac{\Xi_{I,k}}{\delta^2} \chi(\BL,\FR)+\chi(\BL,\BF_k))}{M(1+\mu_k)^2}
\\
&\times
 \Upsilon_{k}(\BK) + \frac{(\omega_k-\frac{\Xi_{I,k}}{\delta})}{M\delta^2(1+\mu_k)} \chi(\BL,\FR)\Gamma(\BK,\FR).\numberthis
\end{align*}
By letting $\BL=\BF_k$ in~(\ref{QKQL}) and adding to~(\ref{eva_zeta_k}), we can construct the system of equations
\begin{equation}
\label{PI_Eq}
\bm{\Pi}\bm{\Upsilon}(\BK)= \bm{\chi}(\BK)+\bm{\varepsilon},
\end{equation}
with $\|\bm{\varepsilon} \|=o(1)$ and $\bm{\Upsilon}(\BK)=[{\Upsilon}_1(\BK),{\Upsilon}_2(\BK),...,\Gamma(\FR,\BK)]^{T}$. We can obtain $\Upsilon_k(\BL)$ by solving~(\ref{PI_Eq}). By the same approach above, we can also construct the equations for $\Lambda_{k,l}$ as
\begin{equation}
\label{lambda_eq}
\Lambda_{k,l}=\sum_{m=1}^{K} \frac{\Xi_{k,m}\Lambda_{m,l}}{L(1+\mu_m)^2}   +  \frac{L\Bxi_{I}}{M\delta^2}\Upsilon_{l}(\FR)+\frac{L}{M}[\BXi \Upsilon(\BF_l)]_{k}+  o(1),
\end{equation}
such that $\Lambda_{k,l}$ can be obtained by solving~(\ref{lambda_eq}) and using the evaluation of  $\Upsilon_k(\BL)$.
\end{proof}

\appendices

\ifCLASSOPTIONcaptionsoff
  \newpage
\fi

\bibliographystyle{IEEEtran}

\bibliography{IEEEabrv,ref}

\end{document}